\documentclass[11pt]{article}%
\pdfoutput=1
\usepackage{amsfonts}
\usepackage{amsmath}
\usepackage{fullpage}
\usepackage{amssymb}
\usepackage{graphicx}
\usepackage{hyperref}%
\usepackage{color}
\providecommand{\U}[1]{\protect\rule{.1in}{.1in}}
\newtheorem{theorem}{Theorem}

\newtheorem{definition}[theorem]{Definition}

\newtheorem{lemma}[theorem]{Lemma}

\newtheorem{proposition}[theorem]{Proposition}

\newenvironment{proof}[1][Proof]{\noindent\textbf{#1.} }{\ \rule{0.5em}{0.5em}}
\numberwithin{equation}{section}

\begin{document}

\title{Bounds on entanglement distillation and secret key agreement for
quantum broadcast channels}

\author{Kaushik P. Seshadreesan\thanks{Hearne Institute for Theoretical Physics, Department of Physics and
Astronomy, Louisiana State University, Baton Rouge, Louisiana 70803,
USA} \and Masahiro Takeoka\thanks{National Institute of Information and Communications Technology, Koganei,
Tokyo 184-8795, Japan} \and Mark M.~Wilde\footnotemark[1] \thanks{Center for Computation and Technology, Louisiana State University,
Baton Rouge, Louisiana 70803, USA}}
\maketitle
\begin{abstract}
The squashed entanglement of a quantum channel is an additive function
of quantum channels, which finds application as an upper bound on
the rate at which secret key and entanglement can be generated when
using a quantum channel a large number of times in addition to unlimited
classical communication. This quantity has led to an upper bound of
$\log((1+\eta)/(1-\eta))$ on the capacity of a pure-loss bosonic
channel for such a task, where $\eta$ is the average fraction of
photons that make it from the input to the output of the channel.
The purpose of the present paper is to extend these results beyond
the single-sender single-receiver setting to the more general case
of a single sender and multiple receivers (a quantum broadcast channel).
We employ multipartite generalizations of the squashed entanglement
to constrain the rates at which secret key and entanglement can be
generated between any subset of the users of such a channel, along
the way developing several new properties of these measures. We apply
our results to the case of a pure-loss broadcast channel with one sender
and two receivers.
\end{abstract}

\section{Introduction}

Quantum key distribution (QKD) refers to the quantum communication
task of generating a shared secret key between two or more cooperating
parties that is information-theoretically secure against an all-powerful
eavesdropper \cite{BB84,SBCDLP09}. The security of QKD is based on
physical principles, guaranteed by the laws of quantum mechanics.
QKD protocols such as BB84 \cite{BB84}, Ekert91 \cite{Eke91} and
CVGG02 \cite{GG02} have been studied both theoretically and experimentally
over many years since the original BB84 proposal.

Practical implementations of QKD over point-to-point fibre optical
links are known to suffer from an exponential decay of the secret
key rate with increasing distance of communication. Recently, it has
been proven mathematically that key distillation over a pure-loss bosonic
channel is fundamentally constrained by such a rate-loss
trade-off \cite{TGW14IEEE,TGW14Nat}, which cannot be circumvented unless
augumented by the use of quantum repeaters \cite{Scherer2011,Lvovsky2009},
for which we do not yet have an operational implementation. The primary
mathematical tool used in \cite{TGW14IEEE,TGW14Nat} to establish
this result is an entanglement measure known as the squashed
entanglement \cite{CW04}. (See also the related works \cite{Tucc99,Tucc02}.)
The squashed entanglement
possesses nearly all the desired properties of an entanglement measure
\cite{CW04,AF04,KW04,BCY11}. From this measure, one can construct
a function of a channel known as the squashed entanglement of the
channel \cite{TGW14IEEE}, which is defined as the maximum squashed
entanglement that can be generated between the input and output of
the channel. The idea behind these quantities stems from classical
information theory, being inspired by the intrinsic information
\cite{MW99}, which can be seen by tracing their roots to earlier
work in \cite{C02}.

Going forward from the results of \cite{TGW14IEEE,TGW14Nat}, a natural
question of interest is to determine rate-loss trade-offs in a multi-user
setting, e.g., in a setting with one sender and multiple receivers.
While one possibility is to consider optical switches or wavelength-division
multiplexing between the transmitter and each of the receiver nodes,
such an architecture would be prohibitively expensive, owing to the
need to have a full QKD system for each node. Alternatively, architectures
based on point-to-multipoint links, modeled as quantum broadcast channels,
have been suggested to accomplish secure multinode networks \cite{Tow97}.
Indeed we cannot hope to circumvent the already established rate-loss
trade-offs when going to these settings, simply because one could
always obtain an upper bound on the achievable rates in such a setting
by grouping all receivers together as a single receiver. Nevertheless, we can hope to refine
our understanding of the rate-loss trade-off. The main purpose of
the present paper is to accomplish just that for a quantum broadcast
channel connecting a single sender to an arbitrary number of receivers.

The secret-key agreement capacity of a noisy quantum channel is defined
as the highest rate at which arbitrarily secure secret-key bits can
be generated by using the channel an arbitrarily large number of times
in addition to unlimited forward and backward classical communication.
This capacity was first considered in the classical context in \cite{Mau93,AC93}
and later on in the more general quantum context. Similarly, a related quantity---the entanglement
distillation capacity of a channel, is defined as the highest rate at which entanglement
can be generated using the channel many times. Note that both of these
definitions include ``direct'' communication: the secret key generated
combined with unlimited classical communication allows for secure
communication via the one-time pad protocol and the entanglement generated
combined with the classical communication allows for quantum communication
via the teleportation protocol \cite{PhysRevLett.70.1895}.

In this work, we establish constraints on the secret-key agreement
capacity and the entanglement distillation capacity of a quantum broadcast channel. Our bounds are based on multipartite generalizations of the squashed entanglement \cite{YHHHOS09,AHS08}
and constrain the rates at which secret key or entanglement can be
established between any subset of the users of a quantum broadcast
channel. For an example, see Theorem~\ref{thm:two receivers} for
our upper bounds on a single-sender two-receiver broadcast channel.
It should be noted that the capacity of the quantum broadcast channel
has been explored in several contexts, including classical communication
\cite{GS07,GSE07,RSW14,SW11} and private and quantum communication
\cite{DHL10,YHD2006}. However, prior to our work, no nontrivial upper
bound had been established for the secret-key agreement and entanglement
distillation capacity of a quantum broadcast channel assisted by unlimited
classical communication between all parties.

The paper is organized as follows. We begin by recalling some preliminary
notions in Section~\ref{sec:Preliminaries}. As part of the preliminaries,
we include a brief review of the squashed entanglement and its multipartite
generalizations in Section~\ref{sub:Bi and Mul Sq Ent}. In Section~\ref{sec:auxiliary-lemmas},
we prove some new auxiliary lemmas regarding several multipartite
squashed entanglements, which play important roles in our main theorem.
Following that, we introduce the quantum broadcast channel in Section~\ref{sub:Quantum-broadcast-channels}
and describe a protocol for secret-key agreement and entanglement
distillation over such a channel. In Section~\ref{sec:Two Receiver UB},
we give upper bounds on the achievable secret-key agreement and entanglement
distillation rates over a single-sender, two-receiver quantum broadcast
channel. Following that, in Section~\ref{sec:m receiver UB}, we
give our general theorem constraining the rates at which secret-key
agreement and entanglement distillation are possible when using a
quantum broadcast channel with one sender and $m$ receivers. In Section~\ref{sec:BBC},
we apply our results to a single-sender, two-receiver bosonic broadcast
channel. Finally, we summarize with a conclusion.

\section{Preliminaries\label{sec:Preliminaries}}

\subsection{States, systems, channels, and measurements}

\label{sec:states channels}Let $\mathcal{{B(H)}}$ denote
the algebra of bounded linear operators acting on a Hilbert space
$\mathcal{{H}}$. Let $\mathcal{{B(H)}}_{+}$ denote the
subset of positive semi-definite operators. An operator $\rho$ is
a density operator, representing the state of a quantum system, if
$\rho\in\mathcal{{B(H)}}_{+}$ and $\operatorname{Tr}\left\{ \rho\right\} =1$.
Let $\mathcal{{S(H)}}$ denote the set of density operators
acting on $\mathcal{{H}}$. We also use the term ``quantum state\textquotedblright \ or
just ``state\textquotedblright \ interchangeably with the term ``density
operator.\textquotedblright \ The tensor product $\mathcal{{H}}_{A}\otimes\mathcal{{H}}_{B}$
of two Hilbert spaces $\mathcal{H}_{A}$ and $\mathcal{H}_{B}$ is
also denoted by $\mathcal{H}_{AB}$.\ Given a multipartite density
operator $\rho_{AB}\in\mathcal{S}(\mathcal{H}_{AB})$, the reduced
density operator on system $A$ is written in terms of the partial
trace as $\rho_{A}=\ $Tr$_{B}\left\{ \rho_{AB}\right\} $. An extension
of a state $\rho_{A}\in\mathcal{S}(\mathcal{H}_{A})$ is
a state $\Omega_{RA}\in\mathcal{S}(\mathcal{H}_{RA})$
such that Tr$_{R}\left\{ \Omega_{RA}\right\} =\rho_{A}$.

A linear map $\mathcal{N}_{A\rightarrow B}:\mathcal{B}(\mathcal{H}_{A})\rightarrow\mathcal{B}(\mathcal{H}_{B})$\ is
positive if $\mathcal{N}_{A\rightarrow B}(\sigma_{A})\in\mathcal{B}(\mathcal{H}_{B})_{+}$
whenever $\sigma_{A}\in\mathcal{B}(\mathcal{H}_{A})_{+}$.
Let id$_{A}$ denote the identity map acting on a system $A$. A linear
map $\mathcal{N}_{A\rightarrow B}$ is completely positive if the
map id$_{R}\otimes\mathcal{N}_{A\rightarrow B}$ is positive for a
reference system $R$ of arbitrary size. A linear map $\mathcal{N}_{A\rightarrow B}$
is trace-preserving if Tr$\left\{ \mathcal{N}_{A\rightarrow B}\left(\tau_{A}\right)\right\} =\ $Tr$\left\{ \tau_{A}\right\} $
for all input operators $\tau_{A}\in\mathcal{B}(\mathcal{H}_{A})$.
A linear map is a quantum channel if it is both completely positive
and trace-preserving (CPTP). An isometric extension $U_{A\rightarrow BE}^{\mathcal{N}}$
of a channel $\mathcal{N}_{A\rightarrow B}$ acting on a state $\rho_{A}\in\mathcal{S}(\mathcal{H}_{A})$
is a linear map that satisfies the following: 
\begin{equation}
\text{Tr}_{E}\left\{ U_{A\rightarrow BE}^{\mathcal{N}}\rho_{A}(U_{A\rightarrow BE}^{\mathcal{N}})^{\dag}\right\} =\mathcal{N}_{A\rightarrow B}(\rho_{A}),\ \ \ \ \ U_{\mathcal{N}}^{\dagger}U_{\mathcal{N}}=I_{A},\ \ \ \ \ U_{\mathcal{N}}U_{\mathcal{N}}^{\dagger}=\Pi_{BE},
\end{equation}
where $\Pi_{BE}$ is a projection onto a subspace of the Hilbert space
$\mathcal{H}_{B}\otimes\mathcal{H}_{E}$.

A measurement channel is a quantum channel with a quantum input and
a classical output, specified as follows:
\begin{equation}
\omega\rightarrow\sum_{m}\text{Tr}\left\{ \Lambda^{m}\omega\right\} \vert m\rangle \langle m\vert ,
\end{equation}
where $\left\{ \Lambda^{m}\right\} $ is a set of positive semi-definite
operators such that $\sum_{m}\Lambda^{m}=I$ and $\left\{ \vert m\rangle \right\} $
is an orthonormal basis. The set $\left\{ \Lambda^{m}\right\} $ is
also known as a positive operator-valued measure (POVM).

\subsection{Maximally entangled states and GHZ states}

\label{sec:ent states}Along with density operators, we also say that
any unit vector $|\psi\rangle\in\mathcal{H}$ is a quantum state.
Its corresponding density operator is $\vert \psi\rangle \langle \psi\vert $,
and we often make the abbreviation $\psi=\vert \psi\rangle \langle \psi\vert $.
Any bipartite pure state $|\psi\rangle_{AB}\in\mathcal{H}_{AB}$ has
a Schmidt decomposition as follows:
\begin{equation}
\vert \psi\rangle _{AB}\equiv\sum_{i=0}^{d-1}\sqrt{\lambda_{i}}\vert i\rangle _{A}\vert i\rangle _{B},
\end{equation}
where $\{|i\rangle_{A}\}$ and $\{|i\rangle_{B}\}$ form orthonormal
bases in $\mathcal{H}_{A}$ and $\mathcal{H}_{B}$, respectively,
$d$ is the Schmidt rank of the state, $0<\lambda_{i}\leq1$ for all
$i\in\left\{ 0,\ldots,d-1\right\} $, and $\sum_{i=0}^{d-1}\lambda_{i}=1$.
A maximally entangled state of Schmidt rank $d$ is a pure bipartite
state of the following form:
\begin{equation}
\left\vert \Phi\right\rangle _{AB}\equiv\frac{1}{\sqrt{d}}\sum_{i=0}^{d-1}\vert i\rangle _{A}\vert i\rangle _{B},\label{eq:Bi Max Ent}
\end{equation}
and is said to contain $H(A)_{\Phi}\equiv-\operatorname{Tr}\left\{ \Phi_{A}\log\Phi_{A}\right\} =\log d$ entangled bits. In the previous expression, $H$ denotes the von Neumann
entropy of the reduced state $\Phi_{A}=\operatorname{Tr}_{B}\left\{ \Phi_{AB}\right\} $. 

The Greenberger-Horne-Zeilinger (GHZ) state is a multipartite generalization
of the maximally entangled state. An $m$-party GHZ state shared between
systems $A_{1}, \ldots, A_{m}$ is written as
\begin{equation}
\left\vert \Phi\right\rangle _{A_{1}\cdots A_{m}}\equiv\frac{1}{\sqrt{d}}\sum_{i=0}^{d-1}\vert i\rangle _{A_{1}}\otimes\cdots\otimes\vert i\rangle _{A_{m}},\label{eq:GHZ def}
\end{equation}
where $\{\vert i\rangle _{A_{1}}\}$, \ldots{}, $\{\vert i\rangle _{A_{m}}\}$
are orthonormal bases, and is also said to contain $\log d$ entangled bits. Throughout this paper, we refer to GHZ\ states $\left\vert \Phi\right\rangle _{A_{1}\cdots A_{m}}$
as maximally entangled states (however, note that this terminology
depends on which entanglement measure one employs---for the entanglement
measures that we employ in this paper, they are indeed maximally entangled).

\subsection{Bipartite and multipartite private states}

\label{sec:priv states}Let $\gamma_{ABA^{\prime}B^{\prime}}$\ be
a state shared between spatially separated parties Alice and Bob,
such that Alice possesses systems $A$ and $A^{\prime}$ and Bob possesses
systems $B$ and $B^{\prime}$. $\gamma_{ABA^{\prime}B^{\prime}}$
is called a private state \cite{HHHO05,HHHO09} if Alice and Bob can
extract a secret key from it by performing local measurements on $A$
and $B$, which is product with any purifying system of $\gamma_{ABA^{\prime}B^{\prime}}$.
That is, $\gamma_{ABA^{\prime}B^{\prime}}$ is a private state of
$\log d$ private bits if, for any purification $\left\vert \varphi^{\gamma}\right\rangle _{ABA^{\prime}B^{\prime}E}$
of $\gamma_{ABA^{\prime}B^{\prime}}$, the following holds:
\begin{equation}
\left(\mathcal{M}_{A}\otimes\mathcal{M}_{B}\otimes\text{Tr}_{A^{\prime}B^{\prime}}\right)\left(\varphi_{ABA^{\prime}B^{\prime}E}^{\gamma}\right)=\frac{1}{d}\sum_{i=0}^{d-1}\vert i\rangle \langle i\vert _{A}\otimes\vert i\rangle \langle i\vert _{B}\otimes\sigma_{E},
\end{equation}
where $\mathcal{M}(\cdot)=\sum_{i}\vert i\rangle \langle i\vert (\cdot)\vert i\rangle \langle i\vert $
is a von Neumann measurement channel and $\sigma_{E}$ is some state
on the purifying system $E$ (which could depend on the particular
purification). The systems $A^{\prime}$ and $B^{\prime}$ are known
as ``shield systems\textquotedblright \ because they aid in keeping
the key secure from any party possessing the purifying system (part
or all of which might belong to a malicious party). It is a non-trivial
consequence of the above definition that a private state of $\log d$
private bits can be written in the following form \cite{HHHO05,HHHO09}:
\begin{equation}
\gamma_{ABA^{\prime}B^{\prime}}=U_{ABA^{\prime}B^{\prime}}\left(\Phi_{AB}\otimes\rho_{A^{\prime}B^{\prime}}\right)U_{ABA^{\prime}B^{\prime}}^{\dag},\label{eq:private-1}
\end{equation}
where $\Phi_{AB}$ is a maximally entangled state of Schmidt rank
$d$, and 
\begin{equation}
U_{ABA^{\prime}B^{\prime}}=\sum_{i,j=0}^{d-1}\vert i\rangle \langle i\vert _{A}\otimes\vert j\rangle \langle j\vert _{B}\otimes U_{A^{\prime}B^{\prime}}^{ij}
\end{equation}
is a controlled unitary known as a ``twisting unitary.\textquotedblright \ The
advantage of the notion of a private state shared between Alice and
Bob as opposed to a secret key is that there is no need to consider
an eavesdropper in the private state formalism, as where this is necessary
when considering a secret key. We return to this point in Section~\ref{sec: locc pc}.

A multipartite private state is a straightforward generalization of
the bipartite definition \cite{HA06}. Indeed, $\gamma_{A_{1}\cdots A_{m}A_{1}^{\prime}\cdots A_{m}^{\prime}}$\ is
a state of $\log d$ private bits if, for any purification $\left\vert \varphi^{\gamma}\right\rangle _{A_{1}\cdots A_{m}A_{1}^{\prime}\cdots A_{m}^{\prime}E}$
of $\gamma_{A_{1}\cdots A_{m}A_{1}^{\prime}\cdots A_{m}^{\prime}}$,
the following holds:
\begin{equation}
\left(\mathcal{M}_{A_{1}}\otimes\cdots\otimes\mathcal{M}_{A_{m}}\otimes\text{Tr}_{A_{1}^{\prime}\cdots A_{m}^{\prime}}\right)\left(\varphi_{A_{1}\cdots A_{m}A_{1}^{\prime}\cdots A_{m}^{\prime}E}^{\gamma}\right)=\frac{1}{d}\sum_{i=0}^{d-1}\vert i\rangle \langle i\vert _{A_{1}}\otimes\cdots\otimes\vert i\rangle \langle i\vert _{A_{m}}\otimes\sigma_{E},\label{eq:mult-secret-key}
\end{equation}
where $\mathcal{M}$ and $\sigma$ are as before. The above implies
that an $m$-partite private state of $\log d$ private bits is a
quantum state $\gamma_{A_{1}\cdots A_{m}A_{1}^{\prime}\cdots A_{m}^{\prime}}$
that can be written as
\begin{equation}
\gamma_{A_{1}\cdots A_{m}A_{1}^{\prime}\cdots A_{m}^{\prime}}=U_{A_{1}\cdots A_{m}A_{1}^{\prime}\cdots A_{m}^{\prime}}\left(\Phi_{A_{1}\cdots A_{m}}\otimes\rho_{A_{1}^{\prime}\cdots A_{m}^{\prime}}\right)U_{A_{1}\cdots A_{m}A_{1}^{\prime}\cdots A_{m}^{\prime}}^{\dag},\label{eq:mult_private state}
\end{equation}
where $\Phi_{A_{1}\cdots A_{m}}$ is an $m$-qudit maximally entangled
state and 
\begin{equation}
U_{A_{1}\cdots A_{m}A_{1}^{\prime}\cdots A_{m}^{\prime}}=\sum_{i_{1},\ldots, i_{m}=0}^{d-1}\vert i_{1},\ldots,i_{m}\rangle \langle i_{1},\ldots,i_{m}\vert _{A_{1}\cdots A_{m}}\otimes U_{A_{1}^{\prime}\cdots A_{m}^{\prime}}^{i_{1},\ldots,i_{m}}
\end{equation}
is a twisting unitary, where each unitary $U_{A_{1}^{\prime}\cdots A_{m}^{\prime}}^{i_{1},\ldots,i_{m}}$
depends on the values $i_{1},\ldots,i_{m}$.

Note that for brevity of notation, we sometimes suppress the shield
systems when writing a private state. In such a case, it is implicit
that the shield systems are contained within; e.g., the notation $\gamma_{AB}$
implies that system $A$ contains both Alice's share of the key and
a shield and likewise $B$ contains both Bob's share of the key and
a shield.

\subsection{LOCC and LOPC}

\label{sec: locc pc}Local operations and classical communication
(LOCC) is a commonly considered paradigm for distributed quantum information
processing between two or more honest parties \cite{BDSW96}\ (for
a recent discussion, see, e.g., \cite{CLMOW14}). In this paradigm,
$m$ cooperating parties $A_{1},\ldots,A_{m}$ begin by sharing
a quantum state $\rho_{A_{1}\cdots A_{m}}$. They are subsequently
allowed to perform local quantum operations on their own systems and
to communicate with each other using a classical communication channel.
A typical goal of an LOCC\ protocol is to distill a GHZ entangled
state or a multipartite private state.

Local operations and public communication (LOPC) is a related paradigm
that is particularly relevant for quantum key distribution \cite{HHHO05,HHHO09}.
In this paradigm, we have the honest parties $A_{1}, \ldots,
A_{m}$ and an additional untrusted party $E$. All parties begin
by sharing a quantum state $\rho_{A_{1}\cdots A_{m}E}$, and the honest
parties are allowed to perform local quantum operations and public
classical communication, such that all parties have access to the
classical data being communicated over a public classical channel.
The public classical channel is ``authenticated,\textquotedblright \ meaning
that the untrusted party can only learn the classical information
but is not allowed to tamper with it. The usual aim of the trusted
parties in the LOPC\ paradigm is to distill a state that is nearly
indistinguishable from a secret key state of the form in (\ref{eq:mult-secret-key}).

One of the main insights of \cite{HHHO05,HHHO09} was to prove that
the approximate LOCC\ distillation of private states is equivalent
to the approximate LOPC\ distillation of a secret key, when we are
dealing with the ``most paranoid\textquotedblright \ scenario in
which the untrusted party possesses a purifying system of the states
of the honest parties. Thus, this result introduces an important simplification
in which it suffices to focus on the LOCC\ distillation of private
states. Also, we can unify this setting with entanglement distillation,
in which the goal of a given protocol could be to distill both entangled
states and private states at the same time. This is exactly the kind
of scenario that we will consider in this paper.

\subsection{Conditional mutual information and conditional multipartite information}

\label{sec: cmi}Let $\rho_{ABE}$ be a tripartite quantum state on
systems $A$, $B$, and $E$. The quantum conditional mutual information
(QCMI) is defined as
\begin{equation}
I(A;B|E)_{\rho}\equiv H(AE)_{\rho}+H(BE)_{\rho}-H(E)_{\rho}-H(ABE)_{\rho},\label{eq:qcmi}
\end{equation}
where, e.g., $H(AE)_{\rho}\equiv-\operatorname{Tr}\left\{ \rho_{AE}\log\rho_{AE}\right\} $
denotes the von Neumann entropy of the state $\rho_{AE}$, which is
the reduced density operator $\rho_{AE}=\operatorname{Tr}_{B}\left\{ \rho_{ABE}\right\} $.
The QCMI is non-negative, which is a non-trivial fact known as the
strong subadditivity of quantum entropy \cite{LR73,LR73PRL}.
The QCMI is non-increasing under the action of local quantum channels
performed on the systems $A$ and $B$ \cite{CW04}, i.e., 
\begin{equation}
I(A;B|E)_{\rho}\geq I(A^{\prime};B^{\prime}|E)_{\omega},\label{eq:cmi_monotonicity}
\end{equation}
where $\omega_{A^{\prime}B^{\prime}E}\equiv\left(\mathcal{N}_{A\rightarrow A^{\prime}}\otimes\mathcal{M}_{B\rightarrow B^{\prime}}\right)\left(\rho_{ABE}\right)$
with $\mathcal{N}_{A\rightarrow A^{\prime}}$ and $\mathcal{M}_{B\rightarrow B^{\prime}}$
arbitrary local quantum channels performed on the input systems
$A$ and $B$, leading to output systems $A^{\prime}$ and $B^{\prime}$,
respectively. Another interesting property of the QCMI is that for
a four-party pure state $\psi_{ABED}$ it obeys a duality relation
given by $I(A;B|E)_{\psi}=I(A;B|D)_{\psi}$ \cite{DY08,YD09}. The
QCMI finds an operational meaning in the information theoretic task
of quantum state redistribution \cite{DY08}.

For an $m+1$-partite quantum state $\rho_{A_{1}\cdots A_{m}E}$,
there are at least two distinct ways to generalize the conditional
mutual information:
\begin{align}
I(A_{1};\cdots;A_{m}|E)_{\rho} & =\sum_{i=1}^{m}H(A_{i}|E)-H(A_{1}\cdots A_{m}|E)_{\rho},\label{eq:cmmi_I}\\
\widetilde{I}(A_{1};\cdots;A_{m}|E)_{\rho} & =\sum_{i=1}^{m}H(A_{\left[m\right]\backslash\left\{ i\right\} }|E)_{\rho}-\left(m-1\right)H(A_{1}\cdots A_{m}|E)_{\rho}\label{eq:cmmi_S}\\
 & =H(A_{1}\cdots A_{m}|E)_{\rho}-\sum_{i=1}^{m}H(A_{i}|A_{\left[m\right]\backslash\left\{ i\right\} }E)_{\rho},\label{eq:cmmi_S_2}
\end{align}
where the shorthand $A_{\left[m\right]\backslash\left\{ i\right\} }$
indicates all systems $A_{1}\cdots A_{m}$ except for system $A_{i}$.\footnote{In previous work \cite{CMS02}, the quantity $\widetilde{I}(A_{1};\cdots;A_{m}|E)_{\rho}$
was denoted by $S_{m}(A_{1};\cdots;A_{m}|E)_{\rho}$, but
there are at least two difficulties with using this notation. First
and foremost, the letter $S$ is widely used in quantum physics to
denote entropy or uncertainty, while the measure here is not a measure
of uncertainty but rather of correlations. Second, having the subscript
$m$ limits the extension of the notation to the more general scenarios
considered in this paper (c.f., Section~\ref{sec:partition-notation}).} The former is the conditional version of a quantity known as the
total correlation \cite{Watan60} and has been used in a variety of
contexts \cite{PHH08,YHW08,W14}, while the latter was introduced
in \cite{CMS02} and employed later on in \cite{YHHHOS09,YHW08}.
The above two quantities are generally incomparable, but related by the following
formula \cite{YHHHOS09}:
\begin{equation}
I(A_{1};\cdots;A_{m}|E)_{\rho}+\widetilde{I}(A_{1};\cdots;A_{m}|E)_{\rho}=\sum_{i=1}^{m}I(A_{i};A_{\left[m\right]\backslash\left\{ i\right\} }|E)_{\rho}.\label{eq:cmmi_dual}
\end{equation}
For a state $\rho_{BA_{1}A_{2}\cdots A_{m}E}$, the above conditional
multipartite informations obey the following chain rules, respectively
\cite[Section III]{YHHHOS09}: 
\begin{align}
I(BA_{1};\cdots;A_{m}|E)_{\rho} & =I(A_{1};\cdots;A_{m}|BE)_{\rho}+\sum_{i=2}^{m}I(B;A_{i}|E)_{\rho},\label{eq:I_Chain}\\
\widetilde{I}(BA_{1};A_{2}\cdots;A_{m}|E)_{\rho} & =\widetilde{I}(A_{1};A_{2};\cdots;A_{m}|BE)_{\rho}+I(B;A_{2}\cdots A_{m}|E)_{\rho}.\label{eq:S_Chain}
\end{align}
Also, they are additive with respect to tensor-product states, non-negative,
and monotone non-increasing under local quantum channels acting on systems
$A_{1},\ldots, A_{m}$ \cite[Section III]{YHHHOS09}, i.e.,
\begin{align}
I(A_{1};\cdots;A_{m}|E)_{\rho} & \geq I(A_{1}^{\prime};\cdots;A_{m}^{\prime}|E)_{\omega},\label{eq:cmmi_monotonicity}\\
\widetilde{I}(A_{1};\cdots;A_{m}|E)_{\rho} & \geq\widetilde{I}(A_{1}^{\prime};\cdots;A_{m}^{\prime}|E)_{\omega},
\end{align}
where
\begin{equation}
\omega_{A_{1}^{\prime}\cdots A_{m}^{\prime}E}\equiv\left(\mathcal{N}_{A_{1}\rightarrow A_{1}^{\prime}}^{\left(1\right)}\otimes\cdots\otimes\mathcal{{N}}_{A_{m}\rightarrow A_{m}^{\prime}}^{\left(m\right)}\right)\left(\rho_{A_{1}\cdots A_{m}E}\right),
\end{equation}
with $\mathcal{N}_{A_{i}\rightarrow A_{i}^{\prime}}^{\left(i\right)}$
an arbitrary local quantum channel performed on the input system
$A_{i}$, leading to output system $A_{i}^{\prime}$.

\subsection{Bipartite and multipartite squashed entanglement}

\label{sub:Bi and Mul Sq Ent}We begin this section by recalling the
definition of the bipartite squashed entanglement \cite{CW04}.

\begin{definition} 
\label{def:Bi_Sq ent}The squashed entanglement
of a bipartite state $\rho_{AB}$ is defined as
\begin{equation}
E_{\operatorname{sq}}(A;B)_{\rho}\equiv\frac{{1}}{2}\inf_{\omega_{ABE}}\left\{ I(A;B|E)_{\omega}:\operatorname{Tr}_{E}\left\{ \omega_{ABE}\right\} =\rho_{AB}\right\} ,
\end{equation}
where the infimum is taken over all possible extensions $\omega_{ABE}$
of $\rho_{AB}$ and $I(A;B|E)_{\omega}$ is the quantum
conditional mutual information of \eqref{eq:qcmi}. 
\end{definition}

The squashed entanglement possesses many of the properties that are
desired of an entanglement measure. For example, it is monotone non-increasing
under LOCC, additive with respect to tensor-product states, and subadditive
in general \cite{CW04}. It is a faithful entanglement measure, in
the sense that it is equal to zero if and only if the state is separable
\cite{BCY11,LW14}. It is also asymptotically continuous~\cite{AF04}.
The squashed entanglement is normalized on maximally entangled states
and private states: for a maximally entangled state of Schmidt rank
$d$, the squashed entanglement equals $\log d$ \cite{CW04}, and
it is at least $\log d$ for a private state of $\log d$ private
bits \cite[Proposition 4.19]{Chr06}. Furthermore, the squashed entanglement
of a state $\rho_{AB}$ is an upper bound on the rate at which Bell
states or private states can be distilled per copy of $\rho_{AB}$
when using LOCC \cite{CW04,CEHHOR07}.

There are at least two different multipartite generalizations of the
squashed entanglement \cite{YHHHOS09,AHS08}:

\begin{definition} \label{def:Mult_Sq ent}For an $m$-partite quantum
state $\rho_{A_{1}\cdots A_{m}}$, the squashed entanglement measures
$E_{\operatorname{sq}}$ and $\widetilde{E}_{\operatorname{sq}}$
are defined as
\begin{align}
E_{\operatorname{sq}}(A_{1};\cdots;A_{m})_{\rho} & \equiv\frac{{1}}{2}\inf_{\omega_{A_{1}A_{2}\cdots A_{m}E}}\left\{ I(A_{1};\cdots;A_{m}|E)_{\omega}:\operatorname{Tr}_{E}\left\{ \omega_{A_{1}\cdots A_{m}E}\right\} =\rho_{A_{1}\cdots A_{m}}\right\} ,\label{eq:I_Squashed}\\
\widetilde{E}_{\operatorname{sq}}(A_{1};\cdots;A_{m})_{\rho} & \equiv\frac{{1}}{2}\inf_{\omega_{A_{1}A_{2}\cdots A_{m}E}}\left\{ \widetilde{I}(A_{1};\cdots;A_{m}|E)_{\omega}:\operatorname{Tr}_{E}\left\{ \omega_{A_{1}\cdots A_{m}E}\right\} =\rho_{A_{1}\cdots A_{m}}\right\} ,\label{eq:S_Squashed}
\end{align}
where the infima are taken over all possible extensions $\omega_{A_{1}\cdots A_{m}E}$
of $\rho_{A_{1}\cdots A_{m}}$, and $I$ and $\widetilde{I}$ are
the quantum conditional multipartite information quantities given
in \eqref{eq:cmmi_I} and \eqref{eq:cmmi_S}, respectively. \end{definition}

The squashed entanglements defined above have the following alternative
characterization in terms of a ``squashing channel,\textquotedblright \ which
follows from the same reasoning that justifies \cite[Eq.~(3)]{CW04}:

\begin{lemma} Let $\left\vert \varphi^{\rho}\right\rangle _{A_{1}\cdots A_{m}E}$
be a purification of $\rho_{A_{1}\cdots A_{m}}$. Then
\begin{align}
E_{\operatorname{sq}}(A_{1};\cdots;A_{m})_{\rho} & \equiv\frac{{1}}{2}\inf_{\mathcal{S}_{E\rightarrow E^{\prime}}}I(A_{1};\cdots;A_{m}|E^{\prime})_{\omega},\\
\widetilde{E}_{\operatorname{sq}}(A_{1};\cdots;A_{m})_{\rho} & \equiv\frac{{1}}{2}\inf_{\mathcal{S}_{E\rightarrow E^{\prime}}}\widetilde{I}(A_{1};\cdots;A_{m}|E^{\prime})_{\omega},
\end{align}
where the infima are over all squashing channels $\mathcal{S}_{E\rightarrow E^{\prime}}$
and
\begin{equation}
\omega_{A_{1}\cdots A_{m}E^{\prime}}\equiv\mathcal{S}_{E\rightarrow E^{\prime}}\!\left(\varphi_{A_{1}\cdots A_{m}E}^{\rho}\right).
\end{equation}

\end{lemma}

The above multipartite squashed entanglements are both monotone non-increasing
under LOCC, additive with respect to tensor-product states, subadditive
in general, and asymptotically continuous \cite[Section IV]{YHHHOS09}.
They both reduce to the bipartite squashed entanglement from Definition~\ref{def:Bi_Sq ent}
when $m=2$. They also satisfy the following lemmas:

\begin{lemma}{\cite[Section IV-A]{YHHHOS09}}\label{lem:sqent} For
a classical-quantum state
\begin{equation}
\rho_{XAB_{1}\cdots B_{m}}=\sum_{x}p_{X}(x)\vert x\rangle \langle x\vert _{X}\otimes\rho_{AB_{1}\cdots B_{m}}^{x},
\end{equation}
the squashed entanglement measures of Definition~\ref{def:Mult_Sq ent}
satisfy the following property:
\begin{align}
E_{\operatorname{sq}}(AX;B_{1};\cdots;B_{m})_{\rho} & =\sum_{x}p(x)E_{\operatorname{sq}}(A;B_{1};\cdots;B_{m})_{\rho^{x}},\label{eq:pseudo conv I sq}\\
\widetilde{E}_{\operatorname{sq}}(AX;B_{1};\cdots;B_{m})_{\rho} & =\sum_{x}p(x)\widetilde{E}_{\operatorname{sq}}(A;B_{1};\cdots;B_{m})_{\rho^{x}}.\label{eq:pseudo conv S Sq}
\end{align}

\end{lemma}

\begin{lemma} {\cite[Observation~1]{YHHHOS09}}\label{lem:perfect-ent-key} Let
$\Phi_{A_{1}\cdots A_{m}}$ be a maximally entangled GHZ\ state of
Schmidt rank $d$. Then
\begin{equation}
E_{\operatorname{sq}}(A_{1};\cdots;A_{m})_{\Phi}=\widetilde{E}_{\operatorname{sq}}(A_{1};\cdots;A_{m})_{\Phi}=\frac{m}{2}\log d.
\end{equation}
Let $\gamma_{A_{1}\cdots A_{m}}$ be a private state, such that each
key system has dimension $d$. Then
\begin{equation}
\min\left\{ E_{\operatorname{sq}}(A_{1};\cdots;A_{m})_{\gamma},\widetilde{E}_{\operatorname{sq}}(A_{1};\cdots;A_{m})_{\gamma}\right\} \geq\frac{m}{2}\log d.
\end{equation}

\end{lemma}

\subsection{Shorthand for multipartite information measures in terms of partitions}

\label{sec:partition-notation}A partition $\mathcal{G}$\ of a set
$\mathcal{K}$ is a set of non-empty subsets of $\mathcal{K}$ such
that 
\begin{equation}
\bigcup\limits _{\mathcal{X}\in\mathcal{G}}\mathcal{X}=\mathcal{K}
\end{equation}
and for all $\mathcal{X}_{1},\mathcal{X}_{2}\in\mathcal{G}$, 
\begin{equation}
\mathcal{X}_{1}\cap\mathcal{X}_{2}=\emptyset.
\end{equation}
For example, if $\mathcal{K}=\left\{ A,B,C\right\} $, then $\mathcal{G=}\left\{ \left\{ A\right\} ,\left\{ B,C\right\} \right\} $\ is
a partition of $\mathcal{K}$.

The power set $\mathcal{P}(\mathcal{K})$ of a set $\mathcal{K}$
is the set of all subsets of $\mathcal{K}$. Let $\mathcal{P}_{\geq1}(\mathcal{K})$
and $\mathcal{P}_{\geq2}(\mathcal{K})$ denote the set
of all non-empty subsets of $\mathcal{K}$ and the set of all non-empty
and non-singleton subsets of $\mathcal{K}$, respectively. For example,
if $\mathcal{K}=\left\{ A,B,C\right\} $, then 
\begin{align}
\mathcal{P}(\mathcal{K}) & =\left\{ \emptyset,\left\{ A\right\} ,\left\{ B\right\} ,\left\{ C\right\} ,\left\{ A,C\right\} ,\left\{ A,B\right\} ,\left\{ B,C\right\} ,\left\{ A,B,C\right\} \right\} ,\\
\mathcal{P}_{\geq1}(\mathcal{K}) & =\left\{ \left\{ A\right\} ,\left\{ B\right\} ,\left\{ C\right\} ,\left\{ A,C\right\} ,\left\{ A,B\right\} ,\left\{ B,C\right\} ,\left\{ A,B,C\right\} \right\} ,\\
\mathcal{P}_{\geq2}(\mathcal{K}) & =\left\{ \left\{ A,C\right\} ,\left\{ A,B\right\} ,\left\{ B,C\right\} ,\left\{ A,B,C\right\} \right\} .\label{eq:Pgeq2}
\end{align}

Let $\mathcal{S}$ be a set, and let $\omega_{\mathcal{S}}$ be an
$\left\vert \mathcal{S}\right\vert $-partite state shared among parties
specified by the elements of $\mathcal{S}$. Let $\mathcal{G}$ be
a partition of $\mathcal{S}$. Then we use the shorthand 
\begin{equation}
E_{\operatorname{sq}}(\mathcal{G})_{\omega}
\end{equation}
to denote a multipartite squashed entanglement with grouping of parties
according to the partition$~\mathcal{G}$. For example, if $\omega\equiv\omega_{ABC}$,
$\mathcal{S}=\left\{ A,B,C\right\} $ and $\mathcal{G}_{1}=\left\{ \left\{ A\right\} ,\left\{ B,C\right\} \right\} $,
then 
\begin{equation}
E_{\operatorname{sq}}(\mathcal{G}_{1})_{\omega}=E_{\operatorname{sq}}(A;BC)_{\omega}.
\end{equation}
Similarly, if $\mathcal{G}_{2}=\left\{ \left\{ A\right\} ,\left\{ B\right\} ,\left\{ C\right\} \right\} $,
then 
\begin{equation}
E_{\operatorname{sq}}(\mathcal{G}_{2})_{\omega}=E_{\operatorname{sq}}(A;B;C)_{\omega}.
\end{equation}
We also employ similar shorthands for $I$, $\widetilde{I}$, and
$\widetilde{E}_{\operatorname{sq}}$.

\section{Auxiliary lemmas for the multipartite squashed entanglements}

\label{sec:auxiliary-lemmas}

\subsection{Subadditivity}

The following lemma is a multipartite generalization of \cite[Theorem 3]{TGW14IEEE},
which was one of the main tools used to prove that the squashed entanglement
of a quantum channel is an upper bound on its quantum capacity or
private capacity when assisted by forward and backward classical communication.
Naturally, the following lemma will be one of the important tools
used to prove the main result in this paper.

\begin{lemma}[Subadditivity]\label{lem:TGW Mult} For a $\left(2m+3\right)$-partite
pure quantum state $\psi_{SP_{1}\cdots P_{m}Q_{1}\cdots Q_{m}E_{1}E_{2}},$
the following subadditivity inequalities hold 
\begin{align}
E_{\operatorname{sq}}(S;P_{1}Q_{1};\cdots;P_{m}Q_{m})_{\psi} & \leq E_{\operatorname{sq}}(SQ_{\left[m\right]}E_{2};P_{1};\cdots;P_{m})_{\psi}+E_{\operatorname{sq}}(SP_{\left[m\right]}E_{1};Q_{1};\cdots;Q_{m})_{\psi},\label{eq:mult_TGW_1}\\
\widetilde{E}_{\operatorname{sq}}(S;P_{1}Q_{1};\cdots;P_{m}Q_{m})_{\psi} & \leq\widetilde{E}_{\operatorname{sq}}(SQ_{\left[m\right]}E_{2};P_{1};\cdots;P_{m})_{\psi}+\widetilde{E}_{\operatorname{sq}}(SP_{\left[m\right]}E_{1};Q_{1};\cdots;Q_{m})_{\psi}.\label{eq:mult_TGW_2}
\end{align}

\end{lemma}

\begin{proof} Let
\begin{align}
\tau_{SP_{1}\cdots P_{m}Q_{1}\cdots Q_{m}E_{1}^{\prime}E_{2}} & =\mathcal{{S}}_{E_{1}\rightarrow E_{1}^{\prime}}(\psi_{SP_{1}\cdots P_{m}Q_{1}\cdots Q_{m}E_{1}E_{2}}),\\
\sigma_{SP_{1}\cdots P_{m}Q_{1}\cdots Q_{m}E_{1}E_{2}^{\prime}} & =\mathcal{{S}}_{E_{2}\rightarrow E_{2}^{\prime}}(\psi_{SP_{1}\cdots P_{m}Q_{1}\cdots Q_{m}E_{1}E_{2}}),\\
\omega_{SP_{1}\cdots P_{m}Q_{1}\cdots Q_{m}E_{1}^{\prime}E_{2}^{\prime}} & =\left(\mathcal{{S}}_{E_{1}\rightarrow E_{1}^{\prime}}\otimes\mathcal{{S}}_{E_{2}\rightarrow E_{2}^{\prime}}\right)(\psi_{SP_{1}\cdots P_{m}Q_{1}\cdots Q_{m}E_{1}E_{2}}),
\end{align}
where each $\mathcal{{S}}_{E_{i}\rightarrow E_{i}^{\prime}}$ is an
arbitrary local squashing channel. Let
\begin{equation}
\left\vert \phi^{\omega}\right\rangle _{SP_{1}\cdots P_{m}Q_{1}\cdots Q_{m}E_{1}^{\prime}E_{2}^{\prime}R}
\end{equation}
be a purification of $\omega$ with purifying system $R$.

We first prove $\left(\ref{eq:mult_TGW_1}\right)$. Consider the following
chain of inequalities:
\begin{align}
 & 2E_{\operatorname{sq}}(S;P_{1}Q_{1};\cdots;P_{m}Q_{m})_{\psi}\nonumber \\
 & \leq I(S;P_{1}Q_{1};\cdots;P_{m}Q_{m}\left\vert E_{1}^{\prime}E_{2}^{\prime}\right.)_{\omega}\\
 & =\sum_{i=1}^{m}H(P_{i}Q_{i}\left\vert E_{1}^{\prime}E_{2}^{\prime}\right.)_{\omega}-H\left(P_{1}\cdots P_{m}Q_{1}\cdots Q_{m}\left\vert SE_{1}^{\prime}E_{2}^{\prime}\right.\right)_{\omega}\\
 & =\sum_{i=1}^{m}H(P_{i}Q_{i}\left\vert E_{1}^{\prime}E_{2}^{\prime}\right.)_{\phi}+H\left(P_{1}\cdots P_{m}Q_{1}\cdots Q_{m}\left\vert R\right.\right)_{\phi}\\
 & \leq\sum_{i=1}^{m}\left[H\left(P_{i}\left\vert E_{1}^{\prime}\right.\right)_{\phi}+H\left(Q_{i}\left\vert E_{2}^{\prime}\right.\right)_{\phi}\right]+H\left(P_{1}\cdots P_{m}\left\vert R\right.\right)_{\phi}+H\left(Q_{1}\cdots Q_{m}\left\vert R\right.\right)_{\phi}
\end{align}
The first inequality follows from Definition~\ref{def:Mult_Sq ent}.
The first equality follows from the definition of the quantum conditional
multipartite information. The second equality follows from the duality
of conditional entropy, namely, for a tripartite pure state $\phi_{KLM}$,
$H\left(K\left\vert L\right.\right)_{\phi}=-H\left(K\left\vert M\right.\right)_{\phi}$.
The second inequality is a consequence of the strong subadditivity
of quantum entropy $I\left(K;L\left\vert M\right.\right)\geq0$. Continuing
from above,
\begin{align}
 & =\sum_{i=1}^{m}H\left(P_{i}\left\vert E_{1}^{\prime}\right.\right)_{\omega}-H\left(P_{1}\cdots P_{m}\left\vert SQ_{1}\cdots Q_{m}E_{1}^{\prime}E_{2}^{\prime}\right.\right)_{\omega}\nonumber \\
 & \ \ \ \ \ \ \ \ \ \ +\sum_{i=1}^{m}H\left(Q_{i}\left\vert E_{2}^{\prime}\right.\right)_{\omega}-H\left(Q_{1}\cdots Q_{m}\left\vert SP_{1}\cdots P_{m}E_{1}^{\prime}E_{2}^{\prime}\right.\right)_{\omega}\\
 & =I\left(SQ_{1}\cdots Q_{m}E_{2}^{\prime};P_{1};\cdots;P_{m}\left\vert E_{1}^{\prime}\right.\right)_{\omega}+I\left(SP_{1}\cdots P_{m}E_{1}^{\prime};Q_{1};\cdots;Q_{m}\left\vert E_{2}^{\prime}\right.\right)_{\omega}\\
 & \leq I\left(SQ_{1}\cdots Q_{m}E_{2};P_{1};\cdots;P_{m}\left\vert E_{1}^{\prime}\right.\right)_{\tau}+I\left(SP_{1}\cdots P_{m}E_{1};Q_{1};\cdots;Q_{m}\left\vert E_{2}^{\prime}\right.\right)_{\sigma}
\end{align}
The first equality follows from the duality of conditional entropy
and the second from rewriting the linear sum of conditional entropies
in terms of a multipartite conditional mutual information. The final
inequality follows from the data processing inequality for the quantum
conditional multipartite information. Since the above calculations
are independent of the choice of squashing channels $\mathcal{{S}}_{E_{i}\rightarrow E_{i}^{\prime}}$,
and since $E_{1}$ purifies the state on $SP_{1}Q_{1}\cdots P_{m}Q_{m}E_{2}$
and $E_{2}$ purifies the state on $SP_{1}Q_{1}\cdots P_{m}Q_{m}E_{1}$,
the inequality in $\left(\ref{eq:mult_TGW_1}\right)$ follows.

We now prove the inequality in $\left(\ref{eq:mult_TGW_2}\right)$.
The proof idea is similar to the above, but we give it below for completeness.
Consider the following chain of inequalities:
\begin{align}
 & 2\widetilde{E}_{\operatorname{sq}}(S;P_{1}Q_{1};\cdots;P_{m}Q_{m})_{\psi}\nonumber \\
 & \leq\widetilde{I}(S;P_{1}Q_{1};\cdots;P_{m}Q_{m}\left\vert E_{1}^{\prime}E_{2}^{\prime}\right.)_{\omega}\\
 & =H\left(SP_{1}Q_{1}\cdots P_{m}Q_{m}|E_{1}^{\prime}E_{2}^{\prime}\right)_{\omega}-H\left(S|P_{1}Q_{1}\cdots P_{m}Q_{m}E_{1}^{\prime}E_{2}^{\prime}\right)_{\omega}\nonumber \\
 & \ \ \ \ \ \ \ \ \ \ -\sum_{i=1}^{m}H\left(P_{i}Q_{i}|SP_{\left[m\right]\backslash\left\{ i\right\} }Q_{\left[m\right]\backslash\left\{ i\right\} }E_{1}^{\prime}E_{2}^{\prime}\right)_{\omega}\\
 & =H\left(P_{1}Q_{1}\cdots P_{m}Q_{m}|E_{1}^{\prime}E_{2}^{\prime}\right)_{\omega}-\sum_{i=1}^{m}H\left(P_{i}Q_{i}|SP_{\left[m\right]\backslash\left\{ i\right\} }Q_{\left[m\right]\backslash\left\{ i\right\} }E_{1}^{\prime}E_{2}^{\prime}\right)_{\omega}\\
 & =H\left(P_{1}Q_{1}\cdots P_{m}Q_{m}|E_{1}^{\prime}E_{2}^{\prime}\right)_{\omega}+\sum_{i=1}^{m}H\left(P_{i}Q_{i}|R\right)_{\phi}
\end{align}
The first inequality follows from the definition of $\widetilde{E}_{\operatorname{sq}}$.
The first equality follows from expanding $\widetilde{I}$ with (\ref{eq:cmmi_S_2}).
The second equality follows from $H\left(AB|C\right)-H\left(A|BC\right)=H\left(B|C\right)$
with $A\equiv S$, $B\equiv P_{1}Q_{1}\cdots P_{m}Q_{m}$, and $C=E_{1}'E_{2}'$.
The third equality follows from duality of conditional entropy. Continuing
from above,
\begin{align}
 & \leq H\left(P_{1}\cdots P_{m}|E_{1}^{\prime}\right)_{\omega}+H\left(Q_{1}\cdots Q_{m}|E_{2}^{\prime}\right)_{\omega}+\sum_{i=1}^{m}H\left(P_{i}|R\right)_{\phi}+H\left(Q_{i}|R\right)_{\phi}\\
 & =H\left(P_{1}\cdots P_{m}|E_{1}^{\prime}\right)_{\omega}+H\left(Q_{1}\cdots Q_{m}|E_{2}^{\prime}\right)_{\omega}-\sum_{i=1}^{m}H\left(P_{i}|E_{1}^{\prime}E_{2}^{\prime}SP_{\left[m\right]\backslash\left\{ i\right\} }Q_{\left[m\right]}\right)_{\omega}\nonumber \\
 & \ \ \ \ \ \ \ \ \ \ -\sum_{i=1}^{m}H\left(Q_{i}|E_{1}^{\prime}E_{2}^{\prime}SP_{\left[m\right]}Q_{\left[m\right]\backslash\left\{ i\right\} }\right)_{\omega}\\
 & =\widetilde{I}\left(SQ_{1}\cdots Q_{m}E_{2}^{\prime};P_{1};\cdots;P_{m}\left\vert E_{1}^{\prime}\right.\right)_{\omega}+\widetilde{I}\left(SP_{1}\cdots P_{m}E_{1}^{\prime};Q_{1};\cdots;Q_{m}\left\vert E_{2}^{\prime}\right.\right)_{\omega}\\
 & \leq\widetilde{I}\left(SQ_{1}\cdots Q_{m}E_{2};P_{1};\cdots;P_{m}\left\vert E_{1}^{\prime}\right.\right)_{\tau}+\widetilde{I}\left(SP_{1}\cdots P_{m}E_{1};Q_{1};\cdots;Q_{m}\left\vert E_{2}^{\prime}\right.\right)_{\sigma}.
\end{align}
The first inequality follows from several applications of the strong
subadditivity of quantum entropy. The first equality is again duality
of conditional entropy. The final equality is from the definition
of $\widetilde{I}$ in\ (\ref{eq:cmmi_S_2}) and the last inequality
follows from the monotonicity of $\widetilde{I}$ under local quantum
operations. Since the above calculations are independent of the choice
of squashing channels $\mathcal{{S}}_{E_{i}\rightarrow E_{i}^{\prime}}$,
and since $E_{1}$ purifies the state on $SP_{1}Q_{1}\cdots P_{m}Q_{m}E_{2}$
and $E_{2}$ purifies the state on $SP_{1}Q_{1}\cdots P_{m}Q_{m}E_{1}$,
the inequality in $\left(\ref{eq:mult_TGW_2}\right)$ follows. \end{proof}

\subsection{Monotonicity under groupings}

\begin{lemma} \label{lem:Grouping systems}The squashed entanglement
measures of Definition \ref{def:Mult_Sq ent} are non-increasing under
grouping of subsystems, i.e., for a state $\rho_{A_{1}\cdots A_{m}}$,
\begin{align}
E_{\operatorname{sq}}(A_{1};\cdots;A_{m})_{\rho} & \geq E_{\operatorname{sq}}\left(A_{1}A_{2};A_{3;}\cdots;A_{m}\right)_{\rho}\label{eq:I sq grp}\\
\widetilde{E}_{\operatorname{sq}}(A_{1};\cdots;A_{m})_{\rho} & \geq\widetilde{E}_{\operatorname{sq}}\left(A_{1}A_{2};A_{3;}\cdots;A_{m}\right)_{\rho}\label{eq:S sq grp}
\end{align}

\end{lemma}

\begin{proof} Consider the chain rule expansion for $I(A_{1};\cdots;A_{m})_{\rho}$
given in $\left(\ref{eq:cmmi_I}\right)$:
\begin{align}
I\left(A_{1};\cdots;A_{m}\left\vert E\right.\right)_{\rho} & =\sum_{i=1}^{m}H\left(A_{i}\left\vert E\right.\right)_{\rho}-H\left(A_{1}\cdots A_{m}\left\vert E\right.\right)_{\rho}\\
 & =H\left(A_{1}\left\vert E\right.\right)_{\rho}+H\left(A_{2}\left\vert E\right.\right)_{\rho}+\sum_{i=3}^{m}H\left(A_{i}\left\vert E\right.\right)_{\rho}-H\left(A_{1}\cdots A_{m}\left\vert E\right.\right)_{\rho}.\label{eq:S grp manip1-1}
\end{align}
Now consider the same chain rule expansion for $I\left(A_{1}A_{2};A_{3};\cdots;A_{m}\right)_{\rho}$:
\begin{equation}
I\left(A_{1}A_{2};A_{3};\cdots;A_{m}\left\vert E\right.\right)_{\rho}=H\left(A_{1}A_{2}\left\vert E\right.\right)_{\rho}+\sum_{i=3}^{m}H\left(A_{i}\left\vert E\right.\right)_{\rho}-H\left(A_{1}\cdots A_{m}\left\vert E\right.\right)_{\rho}.\label{eq:S grp manip2-1}
\end{equation}
Taking the difference of $\left(\ref{eq:S grp manip1-1}\right)$ and
$\left(\ref{eq:S grp manip2-1}\right)$, we find that 
\begin{align}
 & I\left(A_{1};\cdots;A_{m}\left\vert E\right.\right)_{\rho}-I\left(A_{1}A_{2};A_{3};\cdots;A_{m}\left\vert E\right.\right)_{\rho}\nonumber \\
 & =H\left(A_{1}\left\vert E\right.\right)_{\rho}+H\left(A_{2}\left\vert E\right.\right)_{\rho}-H\left(A_{1}A_{2}\left\vert E\right.\right)_{\rho}\\
 & =I\left(A_{1};A_{2}\left\vert E\right.\right)_{\rho}\\
 & \geq0,
\end{align}
where the last inequality follows from the strong subadditivity of quantum
entropy.

Similarly, consider the chain rule expansion for $\widetilde{I}(A_{1};\cdots;A_{m})$
given in $\left(\ref{eq:cmmi_S}\right)$. We have 
\begin{align}
\widetilde{I}\left(A_{1};\cdots;A_{m}\left\vert E\right.\right)_{\rho} & =\sum_{i=1}^{m}H\left(A_{\left[m\right]\backslash\left\{ i\right\} }\left\vert E\right.\right)_{\rho}-\left(m-1\right)H\left(A_{1}\cdots A_{m}\left\vert E\right.\right)_{\rho}\\
 & =H\left(A_{2}\cdots A_{m}\left\vert E\right.\right)_{\rho}+H\left(A_{1}A_{3}\cdots A_{m}\left\vert E\right.\right)_{\rho}\nonumber \\
 & \ \ \ \ \ \ \ \ \ \ +\sum_{i=3}^{m}H\left(A_{\left[m\right]\backslash\left\{ i\right\} }\left\vert E\right.\right)_{\rho}-\left(m-1\right)H\left(A_{1}\cdots A_{m}\left\vert E\right.\right)_{\rho}.\label{eq:S grp manip1}
\end{align}
Now consider the same chain rule expansion for $\widetilde{I}\left(A_{1}A_{2};A_{3};\cdots;A_{m}\right)$
where we have grouped systems $A_{1}$ and $A_{2}$ into one system.
We have 
\begin{multline}
\widetilde{I}\left(A_{1}A_{2};A_{3};\cdots;A_{m}\left\vert E\right.\right)_{\rho}=H\left(A_{3}\cdots A_{m}\left\vert E\right.\right)_{\rho}\\
+\sum_{i=3}^{m}H\left(A_{\left[m\right]\backslash\left\{ i\right\} }\left\vert E\right.\right)_{\rho}-\left(m-2\right)H\left(A_{1}\cdots A_{m}\left\vert E\right.\right)_{\rho}.\label{eq:S grp manip2}
\end{multline}
Taking the difference of $\left(\ref{eq:S grp manip1}\right)$ and
$\left(\ref{eq:S grp manip2}\right)$, we find that 
\begin{align}
 & \widetilde{I}\left(A_{1};\cdots;A_{m}\left\vert E\right.\right)_{\rho}-\widetilde{I}\left(A_{1}A_{2};A_{3};\cdots;A_{m}\left\vert E\right.\right)_{\rho}\nonumber \\
 & =H\left(A_{2}\cdots A_{m}\left\vert E\right.\right)_{\rho}+H\left(A_{1}A_{3}\cdots A_{m}\left\vert E\right.\right)_{\rho}\nonumber \\
 & \ \ \ \ \ \ \ \ \ \ -H\left(A_{3}\cdots A_{m}\left\vert E\right.\right)_{\rho}-H\left(A_{1}\cdots A_{m}\left\vert E\right.\right)_{\rho}\\
 & =I\left(A_{1};A_{2}\left\vert A_{3}\cdots A_{m}E\right.\right)_{\rho}\\
 & \geq0,
\end{align}
where the last inequality follows from the strong subadditivity of quantum
entropy. The statement of the lemma follows from the above inequalities
and by taking infima. \end{proof}

\subsection{Reduction for product states}

\begin{lemma} \label{lem:prod-state-reduction}Let $\omega_{A_{1}A_{2}\cdots A_{m}}=\rho_{A_{1}}\otimes\sigma_{A_{2}\cdots A_{m}}$,
where $\rho_{A_{1}}$ and $\sigma_{A_{2}\cdots A_{m}}$ are density
operators. Then
\begin{align}
E_{\operatorname{sq}}\left(A_{1};A_{2};\cdots;A_{m}\right)_{\omega} & =E_{\operatorname{sq}}\left(A_{2};A_{3;}\cdots;A_{m}\right)_{\sigma},\\
\widetilde{E}_{\operatorname{sq}}\left(A_{1};A_{2};\cdots;A_{m}\right)_{\omega} & =\widetilde{E}_{\operatorname{sq}}\left(A_{2};A_{3};\cdots;A_{m}\right)_{\sigma}.
\end{align}

\end{lemma}

\begin{proof} We first prove LHS $\geq$ RHS for the inequalities
in the statement of the lemma. Consider that
\begin{align}
E_{\operatorname{sq}}\left(A_{1};A_{2};\cdots;A_{m}\right)_{\omega} & \geq E_{\operatorname{sq}}\left(A_{1}A_{2};\cdots;A_{m}\right)_{\omega}\\
 & \geq E_{\operatorname{sq}}\left(A_{2};\cdots;A_{m}\right)_{\omega}\\
 & =E_{\operatorname{sq}}\left(A_{2};\cdots;A_{m}\right)_{\sigma},
\end{align}
where the first inequality is from monotonicity under groupings (Lemma~\ref{lem:Grouping systems})
and the second inequality is from monotonicity under LOCC. The same
reasoning gives
\begin{equation}
\widetilde{E}_{\operatorname{sq}}\left(A_{1};A_{2};\cdots;A_{m}\right)_{\omega}\geq\widetilde{E}_{\operatorname{sq}}\left(A_{2};A_{3};\cdots;A_{m}\right)_{\sigma}.
\end{equation}
We now prove LHS $\leq$ RHS for the inequalities in the statement
of the lemma. Let $\rho_{A_{1}E}$ extend $\rho_{A_{1}}$ and $\sigma_{A_{2}\cdots A_{m}F}$
extend $\sigma_{A_{2}\cdots A_{m}}$. Then
\begin{align}
2E_{\operatorname{sq}}\left(A_{1};A_{2};\cdots;A_{m}\right)_{\omega} & \leq I\left(A_{1};A_{2};\cdots;A_{m}|EF\right)_{\rho\otimes\sigma}\\
 & =H\left(A_{1}|EF\right)_{\rho\otimes\sigma}+\sum_{i=2}^{m}H\left(A_{i}|EF\right)_{\rho\otimes\sigma}-H\left(A_{1}A_{2}\cdots A_{m}|EF\right)_{\rho\otimes\sigma}\\
 & =H\left(A_{1}|E\right)_{\rho}+\sum_{i=2}^{m}H\left(A_{i}|F\right)_{\sigma}-\left[H\left(A_{1}|E\right)_{\rho}+H\left(A_{2}\cdots A_{m}|F\right)_{\sigma}\right]\\
 & =\sum_{i=2}^{m}H\left(A_{i}|F\right)_{\sigma}-H\left(A_{2}\cdots A_{m}|F\right)_{\sigma}\\
 & =I\left(A_{2};\cdots;A_{m}|F\right)_{\sigma}.
\end{align}
Since the calculation holds independently of the particular extension
of $\sigma$, we can conclude that
\begin{equation}
E_{\operatorname{sq}}\left(A_{1};A_{2};\cdots;A_{m}\right)_{\omega}\leq E_{\operatorname{sq}}\left(A_{2};\cdots;A_{m}\right)_{\sigma}.
\end{equation}
For the other inequality, consider that
\begin{align}
2\widetilde{E}_{\operatorname{sq}}\left(A_{1};A_{2};\cdots;A_{m}\right)_{\omega} & \leq\widetilde{I}\left(A_{1};A_{2};\cdots;A_{m}|EF\right)_{\rho\otimes\sigma}\\
 & =H\left(A_{1}A_{2}\cdots A_{m}|EF\right)_{\rho\otimes\sigma}\nonumber \\
 & \ \ \ \ \ \ \ \ \ \ -\left[H\left(A_{1}|A_{2}\cdots A_{m}EF\right)_{\rho\otimes\sigma}+\sum_{i=2}^{m}H\left(A_{i}|A_{1}A_{\left[2:m\right]\backslash\left\{ i\right\} }EF\right)_{\rho\otimes\sigma}\right]\\
 & =H\left(A_{1}|E\right)_{\rho}+H\left(A_{2}\cdots A_{m}|F\right)_{\rho\otimes\sigma}\nonumber \\
 & \ \ \ \ \ \ \ \ \ \ -\left[H\left(A_{1}|E\right)_{\rho}+\sum_{i=2}^{m}H\left(A_{i}|A_{\left[2:m\right]\backslash\left\{ i\right\} }F\right)_{\sigma}\right]\\
 & =H\left(A_{2}\cdots A_{m}|F\right)_{\rho\otimes\sigma}-\sum_{i=2}^{m}H\left(A_{i}|A_{\left[2:m\right]\backslash\left\{ i\right\} }F\right)_{\sigma}\\
 & =\widetilde{I}\left(A_{2};\cdots;A_{m}|F\right)_{\sigma}.
\end{align}
By the same reasoning as above, we can conclude
\begin{equation}
\widetilde{E}_{\operatorname{sq}}\left(A_{1};A_{2};\cdots;A_{m}\right)_{\omega}\leq\widetilde{E}_{\operatorname{sq}}\left(A_{2};\cdots;A_{m}\right)_{\sigma}.
\end{equation}
This completes the proof.
\end{proof}

\subsection{Multipartite squashed entanglements of maximally entangled states
and private states}

Consider a set $\mathcal{S}=\left\{ A,B,C\right\} $. Let $\Psi_{ABC}$
be a joint state over $A$, $B$, and $C$ of the form
\begin{multline}
\Psi_{ABC} \equiv  \Phi_{A^{\left(1\right)}B^{\left(1\right)}}\otimes\Phi_{A^{\left(2\right)}C^{\left(2\right)}}\otimes\Phi_{B^{\left(3\right)}C^{\left(3\right)}}\otimes\Phi_{A^{\left(4\right)}B^{\left(4\right)}C^{\left(4\right)}}\\
\otimes\gamma_{A^{\left(5\right)}B^{\left(5\right)}}\otimes\gamma_{A^{\left(6\right)}C^{\left(6\right)}}\otimes\gamma_{B^{\left(7\right)}C^{\left(7\right)}}\otimes\gamma_{A^{\left(8\right)}B^{\left(8\right)}C^{\left(8\right)}},\label{eq:example-ideal-state}
\end{multline}
where each $\Phi$ represents a maximally entangled state of the form in \eqref{eq:GHZ def}
and each $\gamma$  represents a private state of the form in \eqref{eq:mult_private state},
and where the quantum system $A$, e.g., has been split into subsystems
$A^{\left(1\right)}A^{\left(2\right)}A^{\left(4\right)}A^{\left(5\right)}A^{\left(6\right)}A^{\left(8\right)}$
to segregate the different kinds of correlations that it holds with
$B$ and $C$. Thus, if $\mathcal{S}$ corresponds to a set of three
parties Alice $A$, Bob $B$, and Charlie $C$, then the state $\Psi_{ABC}$
represents a collection of maximally entangled states and private
states shared between all the non-trivial subsets of the parties, which are
enlisted in the following set:
\begin{equation}
\mathcal{P}_{\geq2}(\mathcal{S})=\left\{ \left\{ A,B\right\} ,\left\{ A,C\right\} ,\left\{ B,C\right\} ,\left\{ A,B,C\right\} \right\} .
\end{equation}
Let us denote the number of entangled bits and private bits shared
between the three parties over the various elements of $\mathcal{{P}}_{\geq2}\mathcal{\left(S\right)}$
by the tuple $\left(E_{AB},E_{AC},E_{BC},E_{ABC},K_{AB},K_{AC},K_{BC},K_{ABC}\right)$
(where $E$ stands for ``Entanglement'' and $K$ for ``Key''). Then, for the
state $\Psi_{ABC}$ in \eqref{eq:example-ideal-state}, we have 
\begin{align}
E_{AB} & \equiv H\left(A^{\left(1\right)}\right)_{\Phi},\ \ \ E_{AC}\equiv H\left(A^{\left(2\right)}\right)_{\Phi},\ \ \ E_{BC}\equiv H\left(B^{\left(3\right)}\right)_{\Phi},\ \ \ E_{ABC}\equiv H\left(A^{\left(4\right)}\right)_{\Phi},\label{rate11}\\
K_{AB} & \equiv H\left(A^{\left(5\right)}\right)_{\gamma},\ \ \ K_{AC}\equiv H\left(A^{\left(6\right)}\right)_{\gamma},\ \ \ K_{BC}\equiv H\left(B^{\left(7\right)}\right)_{\gamma},\ \ \ K_{ABC}\equiv H\left(A^{\left(8\right)}\right)_{\gamma},\label{rate12}
\end{align}
where $H$ denotes the von Neumann entropy. Note that the quantity
$E_{AB}$, which, e.g., is the amount of entanglement shared between Alice
and Bob, is to be understood as the amount of entanglement between the systems
$A^{\left(1\right)}$ and $B^{\left(1\right)}$. Also, for the private
$\gamma$-states it is implicit that we are evaluating the entropies
with respect to the key systems, so that the entropy is equal to the
number of private bits in the state.

Our goal in this section is as follows. For a given state $\Psi_{ABC}$
of the form in \eqref{eq:example-ideal-state}, we want to establish
constraints relating the elements of the tuple $\left(E_{AB},E_{AC},E_{BC},E_{ABC},K_{AB},K_{AC},K_{BC},K_{ABC}\right)$
using the multipartite squashed entanglement quantities discussed
in Section~\ref{sub:Bi and Mul Sq Ent}. For this, we are interested
in determining the multipartite squashed entanglements of $\Psi_{ABC}$
with respect to various nontrivial partitions of $\mathcal{S}=\left\{ A,B,C\right\} $,
which are given by
\begin{align}
\mathcal{G}_{1} & =\left\{ \left\{ A\right\} ,\left\{ B,C\right\} \right\} ,\label{eq:partition-ABC-1}\\
\mathcal{G}_{2} & =\left\{ \left\{ A,B\right\} ,\left\{ C\right\} \right\} ,\\
\mathcal{G}_{3} & =\left\{ \left\{ A,C\right\} ,\left\{ B\right\} \right\} ,\label{eq:partition-ABC-3}\\
\mathcal{G}_{4} & =\left\{ \left\{ A\right\} ,\left\{ B\right\} ,\left\{ C\right\} \right\} .\label{eq:partition-ABC-4}
\end{align}
(Note that we have excluded the trivial partition $\mathcal{G}_{5}=\left\{ \mathcal{S}\right\} $.) 

For partition $\mathcal{G}_{1}$, we obtain
\begin{align}
E_{\operatorname{sq}}(\mathcal{G}_{1})_{\Psi} & =\widetilde{E}_{\operatorname{sq}}(\mathcal{G}_{1})_{\Psi}\\
 & =E_{\operatorname{sq}}\left(A^{\left(1\right)}A^{\left(2\right)}A^{\left(4\right)}A^{\left(5\right)}A^{\left(6\right)}A^{\left(8\right)};B^{\left(1\right)}B^{\left(3\right)}B^{\left(4\right)}B^{\left(5\right)}B^{\left(7\right)}B^{\left(8\right)}C^{\left(2\right)}C^{\left(3\right)}C^{\left(4\right)}C^{\left(6\right)}C^{\left(7\right)}C^{\left(8\right)}\right)_{\Psi}\\
 & =E_{\operatorname{sq}}(A^{\left(1\right)};B^{\left(1\right)})_{\Phi}+E_{\operatorname{sq}}(A^{\left(2\right)};C^{\left(2\right)})_{\Phi}+E_{\operatorname{sq}}(A^{\left(4\right)};B^{\left(4\right)}C^{\left(4\right)})_{\Phi}\nonumber \\
 & \ \ \ \ \ \ +E_{\operatorname{sq}}(A^{\left(5\right)};B^{\left(5\right)})_{\gamma}+E_{\operatorname{sq}}(A^{\left(6\right)};C^{\left(6\right)})_{\gamma}+E_{\operatorname{sq}}(A^{\left(8\right)};B^{\left(8\right)}C^{\left(8\right)})_{\gamma}\\
 & \geq E_{AB}+E_{AC}+E_{ABC}+K_{AB}+K_{AC}+K_{ABC}.\label{eq:G1 E sq 2}
\end{align}
The first equality follows because the two squashed entanglements
are identical in the bipartite case. The second equality follows from
the additivity of squashed entanglement with respect to tensor-product
states and from Lemma~\ref{lem:prod-state-reduction}. The inequality
follows from Lemma~\ref{lem:perfect-ent-key}. A similar line of
reasoning for partitions $\mathcal{G}_{2}$ and $\mathcal{G}_{3}$
yields the following constraints:
\begin{align}
E_{\operatorname{sq}}(\mathcal{G}_{2})_{\Psi} & =\widetilde{E}_{\operatorname{sq}}(\mathcal{G}_{2})_{\Psi}\geq E_{AC}+E_{BC}+E_{ABC}+K_{AC}+K_{BC}+K_{ABC},\\
E_{\operatorname{sq}}\left(\mathcal{G}_{3}\right)_{\Psi} & =\widetilde{E}_{\operatorname{sq}}\left(\mathcal{G}_{3}\right)_{\Psi}\geq E_{AB}+E_{BC}+E_{ABC}+K_{AB}+K_{BC}+K_{ABC}.\label{eq:part-g3}
\end{align}
Finally, for partition $\mathcal{G}_{4}$, we obtain
\begin{align}
E_{\operatorname{sq}}(\mathcal{G}_{4})_{\Psi} & =E_{\operatorname{sq}}\left(A^{\left(1\right)}A^{\left(2\right)}A^{\left(4\right)}A^{\left(5\right)}A^{\left(6\right)}A^{\left(8\right)};B^{\left(1\right)}B^{\left(3\right)}B^{\left(4\right)}B^{\left(5\right)}B^{\left(7\right)}B^{\left(8\right)};C^{\left(2\right)}C^{\left(3\right)}C^{\left(4\right)}C^{\left(6\right)}C^{\left(7\right)}C^{\left(8\right)}\right)_{\Psi}\\
 & =E_{\operatorname{sq}}(A^{\left(1\right)};B^{\left(1\right)})_{\Phi}+E_{\operatorname{sq}}(A^{\left(2\right)};C^{\left(2\right)})_{\Phi}+E_{\operatorname{sq}}(B^{\left(3\right)};C^{\left(3\right)})_{\Phi}+E_{\operatorname{sq}}(A^{\left(4\right)};B^{\left(4\right)};C^{\left(4\right)})_{\Phi}\nonumber \\
 & \ \ \ \ \ \ \ \ \ \ +E_{\operatorname{sq}}(A^{\left(5\right)};B^{\left(5\right)})_{\gamma}+E_{\operatorname{sq}}(A^{\left(6\right)};C^{\left(6\right)})_{\gamma}+E_{\operatorname{sq}}(B^{\left(7\right)};C^{\left(7\right)})_{\gamma}+E_{\operatorname{sq}}(A^{\left(8\right)};B^{\left(8\right)};C^{\left(8\right)})_{\gamma}\\
 & \geq E_{AB}+E_{AC}+E_{BC}+\frac{{3}}{2}E_{ABC}+K_{AB}+K_{AC}+K_{BC}+\frac{{3}}{2}K_{ABC}.\label{eq:G4 E sq}
\end{align}
The second equality follows from the additivity of squashed entanglement
with respect to tensor-product states and from Lemma~\ref{lem:prod-state-reduction}.
The inequality follows from Lemma~\ref{lem:perfect-ent-key}. Similarly,
we also obtain
\begin{equation}
\widetilde{E}_{\operatorname{sq}}(\mathcal{G}_{4})_{\Psi}\geq E_{AB}+E_{AC}+E_{BC}+\frac{{3}}{2}E_{ABC}+K_{AB}+K_{AC}+K_{BC}+\frac{{3}}{2}K_{ABC}.
\end{equation}
Since $\mathcal{G}_{4}$ is a tripartition and the two multipartite
squashed entanglements are not identical in general, we can pick the minimum of
$\widetilde{E}_{\operatorname{sq}}(\mathcal{G}_{4})_{\Psi}$ and $E_{\operatorname{sq}}(\mathcal{G}_{4})_{\Psi}$
to give a tighter upper bound.

The above analysis can be further extended to sets containing more
than three elements. Consider a set of $m+1$ elements, $\mathcal{S}=\left\{ A,B_{1},\ldots,B_{m}\right\} $,
for arbitrary but finite $m$. Let $\Psi_{\mathcal{S}}$ be the following
state:
\begin{equation}
\Psi_{\mathcal{S}}\equiv\bigotimes\limits _{\mathcal{K\in P}_{\geq2}(\mathcal{S})}\Phi_{\mathcal{K}}\otimes\gamma_{\mathcal{K}},\label{eq:big-entangled-state}
\end{equation}
which is a tensor product of all possible entangled states and private
states that could be shared between all subsets of the parties in
$\mathcal{S}$, with it understood that each $\mathcal{K}$ has a
set of distinct subsystems in the tensor product (as is the case for
the example in \eqref{eq:example-ideal-state}). For a given subset
$\mathcal{K\in P}_{\geq2}(\mathcal{S})$, let $E_{\mathcal{K}}$
denote the number of entangled bits (logarithm of the Schmidt rank)\ in
the multiparty GHZ\ entangled state $\Phi_{\mathcal{K}}$, and let
$K_{\mathcal{K}}$ denote the number of private bits in the private
state $\gamma_{\mathcal{K}}$.

\begin{definition} For a given nontrivial partition $\mathcal{G}$
of a set $\mathcal{S}$, we define the set $\mathcal{C}(\mathcal{G})$
of sets by the following procedure. Let $\mathcal{X}_{1},\ldots,\mathcal{X}_{\left\vert \mathcal{G}\right\vert }$
denote all of the sets in the partition $\mathcal{G}$. For each $\mathcal{L}_{\mathcal{X}_{1}}\in\mathcal{P}\left(\mathcal{X}_{1}\right)$,
\ldots{}, $\mathcal{L}_{\mathcal{X}_{\left\vert \mathcal{G}\right\vert }}\in\mathcal{P}\left(\mathcal{X}_{\left\vert \mathcal{G}\right\vert }\right)$,
form the set $\mathcal{L}_{\mathcal{X}_{1}}\cup\cdots\cup\mathcal{L}_{\mathcal{X}_{\left\vert \mathcal{G}\right\vert }}$
and add it to $\mathcal{C}(\mathcal{G})$. At the end,
remove the null set and any singleton sets. \end{definition}

For example, for the partition $\mathcal{G}_{1}=\left\{ \left\{ A\right\} ,\left\{ B,C\right\} \right\} $
of $\mathcal{S}=\left\{ A,B,C\right\} $, this procedure leads to
\begin{equation}
\mathcal{C}(\mathcal{G}_{1})=\left\{ \left\{ A,B\right\} ,\left\{ A,C\right\} ,\left\{ A,B,C\right\} \right\} .\label{eq:CG for G1}
\end{equation}

\begin{definition} For a given nontrivial partition $\mathcal{G}$
of a set $\mathcal{S}$ and an element $\mathcal{M}$ of $\mathcal{{C\left(G\right)}}$,
we define the set $\mathcal{A}\left(\mathcal{M},\mathcal{G}\right)$
as
\begin{equation}
\mathcal{A}\left(\mathcal{M},\mathcal{G}\right)\equiv\left\{ \mathcal{X}\cap\mathcal{M}\ |\ \mathcal{X}\in\mathcal{G}\right\} \backslash\{\emptyset\}.
\end{equation}
 \end{definition}

For example, for the partition $\mathcal{G}_{4}=\left\{ \left\{ A\right\} ,\left\{ B\right\} ,\left\{ C\right\} \right\} $
of $\mathcal{S}=\left\{ A,B,C\right\} $, we have
\begin{equation}
\mathcal{C}(\mathcal{G}_{4})=\left\{ \left\{ A,B\right\} ,\left\{ A,C\right\} ,\left\{ B,C\right\} ,\left\{ A,B,C\right\} \right\} .
\end{equation}
Let us denote the elements of $\mathcal{C}(\mathcal{G}_{4})$
as $\left\{ \mathcal{M}_{1},\mathcal{M}_{2},\mathcal{M}_{3},\mathcal{M}_{4}\right\} .$
Then, we have
\begin{align}
\mathcal{A}(\mathcal{M}_{1},\mathcal{G}_{4}) & =\left\{ \left\{ A\right\} ,\left\{ B\right\} \right\} ,\\
\mathcal{A}(\mathcal{M}_{2},\mathcal{G}_{4}) & =\left\{ \left\{ A\right\} ,\left\{ C\right\} \right\} ,\\
\mathcal{A}(\mathcal{M}_{3},\mathcal{G}_{4}) & =\left\{ \left\{ B\right\} ,\left\{ C\right\} \right\} ,\\
\mathcal{A}(\mathcal{M}_{4},\mathcal{G}_{4}) & =\left\{ \left\{ A\right\} ,\left\{ B\right\} ,\left\{ C\right\} \right\} .
\end{align}

\begin{lemma} \label{lem:max ent state sq ent}Let $\mathcal{S}$
be a set of parties and let $\Psi_{\mathcal{S}}$ be the tensor product
of states defined in \eqref{eq:big-entangled-state}. Then for a given
nontrivial partition $\mathcal{G}$ of $\mathcal{S}$, the squashed
entanglements $E_{\operatorname{sq}}(\mathcal{G})_{\Psi}$
and $\widetilde{E}_{\operatorname{sq}}(\mathcal{G})_{\Psi}$
from Definition \ref{def:Mult_Sq ent} constrain the number of entangled  bits $E_{\mathcal{M}}$
and private bits $K_{\mathcal{M}}$ between the elements of $\mathcal{G}$ as follows:
\begin{equation}
\frac{1}{2}\sum_{\mathcal{M\in C}(\mathcal{G})}\left\vert \mathcal{A}\left(\mathcal{M},\mathcal{G}\right)\right\vert \left(K_{\mathcal{M}}+E_{\mathcal{M}}\right)\leq\min\left\{ E_{\operatorname{sq}}(\mathcal{G})_{\Psi_{\mathcal{S}}},\ \widetilde{E}_{\operatorname{sq}}(\mathcal{G})_{\Psi_{\mathcal{S}}}\right\} .\label{eq: Sq ent formula for max ent}
\end{equation}

\end{lemma}

\begin{proof} Let $\mathcal{G}$ be a nontrivial partition of $\mathcal{S}$.
We begin by considering $E_{\operatorname{sq}}(\mathcal{G})_{\Psi_{\mathcal{S}}}$:
\begin{align}
E_{\operatorname{sq}}(\mathcal{G})_{\Psi_{\mathcal{S}}} & =\sum_{\mathcal{K\in P}_{\geq2}(\mathcal{S})}\left(E_{\operatorname{sq}}(\mathcal{G})_{\Phi_{\mathcal{K}}}+E_{\operatorname{sq}}(\mathcal{G})_{\gamma_{\mathcal{K}}}\right)\label{eq:gen-lower-bound-1}\\
 & =\sum_{\mathcal{M\in C}(\mathcal{G})}\left(E_{\operatorname{sq}}(\mathcal{G})_{\Phi_{\mathcal{M}}}+E_{\operatorname{sq}}(\mathcal{G})_{\gamma_{\mathcal{M}}}\right).
\end{align}
The first equality is a consequence of the additivity of squashed
entanglement with respect to tensor product states. The second equality
follows because $E_{\operatorname{sq}}(\mathcal{G})_{\Phi_{\mathcal{K}}}=E_{\operatorname{sq}}(\mathcal{G})_{\gamma_{\mathcal{K}}}=0$
if $\mathcal{K\subseteq X}$ for some $\mathcal{X\in G}$ and the
algorithm that constructs $\mathcal{C}(\mathcal{G})$ removes
all such $\mathcal{K}$. Now consider that
\begin{align}
2E_{\operatorname{sq}}(\mathcal{G})_{\Phi_{\mathcal{M}}} & =I(\mathcal{G})_{\Phi_{\mathcal{M}}}\\
 & =\sum_{\mathcal{X\in G}}H\left(\mathcal{X}\cap\mathcal{M}\right)_{\Phi_{\mathcal{M}}}-H\left(\mathcal{M}\right)_{\Phi_{\mathcal{M}}}\\
 & =\sum_{\mathcal{X\in G}}H\left(\mathcal{X}\cap\mathcal{M}\right)_{\Phi_{\mathcal{M}}}\\
 & =\left\vert \mathcal{A}\left(\mathcal{M},\mathcal{G}\right)\right\vert E_{\mathcal{M}}.\label{eq:gen-lower-bound-last}
\end{align}
The first equality holds because any pure entangled state is not
extendible, so that any extension system is product with it. The next
equality is from the definition of $I(\mathcal{G})$ (similar
to \eqref{eq:cmmi_I} without the conditioning system). The third
equality follows because the state $\Phi_{\mathcal{M}}$ is pure.
The final equality is a consequence of the definition of the set $\mathcal{A}\left(\mathcal{M},\mathcal{G}\right)$
and the fact that for all $\mathcal{X}\in\mathcal{G}$, where $\mathcal{G}$
is a nontrivial partition of $\mathcal{\mathcal{S}},$ $\mathcal{X}\cap\mathcal{M}\subset\mathcal{M}$
and therefore $H\left(\mathcal{X}\cap\mathcal{M}\right)_{\Phi_{\mathcal{M}}}$
equals the number of entangled bits in $\Phi_{\mathcal{M}}$.

We now prove the following lower bound:
\begin{equation}
E_{\operatorname{sq}}(\mathcal{G})_{\gamma_{\mathcal{M}}}\geq\frac{1}{2}\left\vert \mathcal{A}\left(\mathcal{M},\mathcal{G}\right)\right\vert K_{\mathcal{M}},\label{eq:gen-key-lower-bound}
\end{equation}
which extends Lemma~\ref{lem:perfect-ent-key} and just applies the
idea behind \cite[Proposition~4.19]{Chr06}\ to this more general
case. To do so, let us consider both the key and shield systems of
$\gamma_{\mathcal{M}}$ and label them by $\mathcal{M}$ and $\mathcal{M}^{\prime}$,
respectively, so that we relabel $\gamma_{\mathcal{M}}$ as $\gamma_{\mathcal{MM}^{\prime}}$.
Then $\gamma_{\mathcal{MM}^{\prime}}$ has the following form:
\begin{equation}
\gamma_{\mathcal{MM}^{\prime}}=U_{\mathcal{MM}^{\prime}}\left(\Phi_{\mathcal{M}}\otimes\rho_{\mathcal{M}^{\prime}}\right)U_{\mathcal{MM}^{\prime}}^{\dag},
\end{equation}
where the twisting unitary is
\begin{equation}
U_{\mathcal{MM}^{\prime}}=\sum_{i=0}^{2^{K_{\mathcal{M}}}-1}\vert i\rangle \langle i\vert _{\mathcal{M}}\otimes U_{\mathcal{M}^{\prime}}^{i},
\end{equation}
and where $K_{\mathcal{M}}$ is the number of private bits contained
in $\gamma_{\mathcal{MM}^{\prime}}$. An extension $\gamma_{\mathcal{MM}^{\prime}E}$\ of
$\gamma_{\mathcal{MM}^{\prime}}$ has the following form:
\begin{equation}
\gamma_{\mathcal{MM}^{\prime}E}=U_{\mathcal{MM}^{\prime}}\left(\Phi_{\mathcal{M}}\otimes\rho_{\mathcal{M}^{\prime}E}\right)U_{\mathcal{MM}^{\prime}}^{\dag},\label{eq:gencase private state}
\end{equation}
where $\rho_{\mathcal{M}^{\prime}E}$ is some extension of $\rho_{\mathcal{M}^{\prime}}$.
Let $\gamma_{\mathcal{M}^{\prime}E}^{i}$ denote the following state:
\begin{equation}
\gamma_{\mathcal{M}^{\prime}E}^{i}\equiv U_{\mathcal{M}^{\prime}}^{i}\rho_{\mathcal{M}^{\prime}E}\left(U_{\mathcal{M}^{\prime}}^{i}\right)^{\dag}.
\end{equation}
Then consider that
\begin{align}
H\left(\mathcal{M}\mathcal{M}^{\prime}E\right)_{\gamma} & =H\left(\mathcal{M}^{\prime}E\right)_{\rho}=H\left(\mathcal{M}^{\prime}E\right)_{\gamma^{i}},\\
H(E)_{\gamma} & =H(E)_{\gamma^{i}},\label{eq:complex1}
\end{align}
where we have used some well-known properties of the von Neumann entropy,
namely that it is invariant under unitary transformations, it is additive
on tensor product states and that it is zero for pure states. Then,
for all $i$, we have
\begin{equation}
H\left(\mathcal{M}\mathcal{M}^{\prime}|E\right)_{\gamma}=H\left(\mathcal{M}^{\prime}|E\right)_{\gamma^{i}}.
\end{equation}
This allows us to conclude that
\begin{equation}
H\left(\mathcal{M}\mathcal{M}^{\prime}|E\right)_{\gamma}=\sum_{i=0}^{2^{K_{\mathcal{M}}}-1}\frac{1}{2^{K_{\mathcal{M}}}}H\left(\mathcal{M}^{\prime}|E\right)_{\gamma^{i}},
\end{equation}
where we have simply rewritten the right hand side, since $H\left(\mathcal{M}^{\prime}|E\right)_{\gamma^{i}}$
is the same for all $i$. For all $\mathcal{X\in G}$, we have that
\begin{align}
H\left(\left[\mathcal{X\cap M}\right]\left[\mathcal{X}^{\prime}\cap\mathcal{M}^{\prime}\right]E\right)_{\gamma} & =H\left(\mathcal{X\cap M}\right)_{\gamma}+H\left(\left[\mathcal{X}^{\prime}\cap\mathcal{M}^{\prime}\right]E|\left[\mathcal{X\cap M}\right]\right)_{\gamma}\\
 & =K_{\mathcal{M}}+\sum_{i=0}^{2^{K_{\mathcal{M}}}-1}\frac{1}{2^{K_{\mathcal{M}}}}H\left(\left[\mathcal{X}^{\prime}\cap\mathcal{M}^{\prime}\right]E\right)_{\gamma^{i}},\label{eq:complex2}
\end{align}
where we have used $\mathcal{X}^{\prime}$ to label the shield systems
corresponding to the key systems in $\mathcal{X}$. The first term
in \eqref{eq:complex2} follows because $H\left(\mathcal{X\cap M}\right)_{\gamma}=H\left(\mathcal{X\cap M}\right)_{\Phi_{\mathcal{M}}}$,
the number of entangled bits in $\Phi_{\mathcal{M}}$, which indeed
equals the number of private bits in $\gamma.$ Thus, from \eqref{eq:complex2}
and \eqref{eq:complex1}, we have 
\begin{equation}
H\left(\left[\mathcal{X\cap M}\right]\left[\mathcal{X}^{\prime}\cap\mathcal{M}^{\prime}\right]|E\right)_{\gamma}=K_{\mathcal{M}}+\sum_{i}\frac{1}{2^{K_{\mathcal{M}}}}H\left(\left[\mathcal{X}^{\prime}\cap\mathcal{M}^{\prime}\right]|E\right)_{\gamma^{i}}\label{eq:multiparty-key-squash}
\end{equation}
The multipartite conditional mutual information of $\gamma$ across
the partition $\mathcal{G}$ of the key systems and the analogous
partition $\mathcal{G}^{\prime}$ of the shield systems can thus be
written as
\begin{align}
I(\mathcal{GG}^{\prime}|E)_{\gamma} & =\sum_{\mathcal{X\in G}}H\left(\left[\mathcal{X\cap M}\right]\left[\mathcal{X}^{\prime}\cap\mathcal{M}^{\prime}\right]|E\right)_{\gamma}-H\left(\mathcal{M}\mathcal{M}^{\prime}|E\right)_{\gamma}\label{eq:squash-justify-1}\\
 & =\left\vert \mathcal{A}\left(\mathcal{M},\mathcal{G}\right)\right\vert K_{\mathcal{M}}+\sum_{\mathcal{X\in G}}\sum_{i}\frac{1}{2^{K_{\mathcal{M}}}}H\left(\left[\mathcal{X}^{\prime}\cap\mathcal{M}^{\prime}\right]|E\right)_{\gamma^{i}}\\
 & \ \ \ \ \ \ \ -\sum_{i}\frac{1}{2^{K_{\mathcal{M}}}}H\left(\mathcal{M}^{\prime}|E\right)_{\gamma^{i}}\\
 & =\left\vert \mathcal{A}\left(\mathcal{M},\mathcal{G}\right)\right\vert K_{\mathcal{M}}+\sum_{i}\frac{1}{2^{K_{\mathcal{M}}}}I\left(\mathcal{G}^{\prime}|E\right)_{\gamma^{i}}\\
 & \geq\left\vert \mathcal{A}\left(\mathcal{M},\mathcal{G}\right)\right\vert K_{\mathcal{M}}.\label{eq:squash-justify-last}
\end{align}
The first equality follows from the definition of conditional multipartite
information in \eqref{eq:cmmi_I}. The second equality follows from
(\ref{eq:complex2}) and (\ref{eq:multiparty-key-squash}) and the definition
of the set $\mathcal{A}\left(\mathcal{M},\mathcal{G}\right)$. The
third equality follows once again from the definition of conditional
multipartite information in \eqref{eq:cmmi_I}. Finally, the fourth
inequality in (\ref{eq:squash-justify-last}) from the strong subadditivity of quantum entropy,
namely that $I\left(\mathcal{G}^{\prime}|E\right)_{\gamma^{i}}\geq0$
for any quantum state $\gamma^{i}$. Since the inequality is independent
of the particular extension of $\gamma_{\mathcal{MM}^{\prime}}$,
from the definition of the multipartite squashed entanglement in \eqref{def:Mult_Sq ent},
we can conclude (\ref{eq:gen-key-lower-bound}).\ Putting together
(\ref{eq:gen-lower-bound-1})-(\ref{eq:gen-lower-bound-last}) and
(\ref{eq:gen-key-lower-bound}), we find that
\begin{equation}
E_{\operatorname{sq}}(\mathcal{G})_{\Psi_{\mathcal{S}}}\geq\frac{1}{2}\sum_{\mathcal{M\in C}(\mathcal{G})}\left\vert \mathcal{A}\left(\mathcal{M},\mathcal{G}\right)\right\vert \left(K_{\mathcal{M}}+E_{\mathcal{M}}\right)\label{eq:upper bound1}
\end{equation}

Similarly, we have
\begin{align}
\widetilde{E}_{\operatorname{sq}}(\mathcal{G})_{\Psi_{\mathcal{S}}} & =\sum_{\mathcal{K\in P}_{\geq2}(\mathcal{S})}\left(\widetilde{E}_{\operatorname{sq}}(\mathcal{G})_{\Phi_{\mathcal{K}}}+\widetilde{E}_{\operatorname{sq}}(\mathcal{G})_{\gamma_{\mathcal{K}}}\right)\\
 & =\sum_{\mathcal{M\in C}(\mathcal{G})}\left(\widetilde{E}_{\operatorname{sq}}(\mathcal{G})_{\Phi_{\mathcal{M}}}+\widetilde{E}_{\operatorname{sq}}(\mathcal{G})_{\gamma_{\mathcal{M}}}\right)
\end{align}
The proof is similar to the proof of \eqref{eq:upper bound1}. Using
the definition of $\tilde{I}(\mathcal{G})$ similar to
\eqref{eq:cmmi_S_2} (except for the conditioning system) and that
of the the multipartite squashed entanglement $\widetilde{E}_{\operatorname{sq}}$
in Definition \ref{def:Mult_Sq ent}, and a similar line of reasoning as given
before for (\ref{eq:gen-lower-bound-1})-(\ref{eq:gen-lower-bound-last}),
we obtain 
\begin{align}
2\widetilde{E}_{\operatorname{sq}}(\mathcal{G})_{\Phi_{\mathcal{M}}} & =\widetilde{I}(\mathcal{G})_{\Phi_{\mathcal{M}}}\\
 & =H(\mathcal{S})_{\Phi_{\mathcal{M}}}-\sum_{\mathcal{X\in G}}H\left(\mathcal{X}\cap\mathcal{M}|\left(\mathcal{S}\backslash\mathcal{X}\right)\cap\mathcal{M}\right)_{\Phi_{\mathcal{M}}}\\
 & =\sum_{\mathcal{X\in G}}H\left(\mathcal{X}\cap\mathcal{M}\right)_{\Phi_{\mathcal{M}}}\\
 & =\left\vert \mathcal{A}\left(\mathcal{M},\mathcal{G}\right)\right\vert E_{\mathcal{M}}.
\end{align}
By the same reasoning used to conclude\ (\ref{eq:multiparty-key-squash}),
we can conclude that
\begin{equation}
H\left(\left[\mathcal{M}\backslash\left[\mathcal{X\cap M}\right]\right]\left[\mathcal{M}^{\prime}\backslash\left[\mathcal{X}^{\prime}\mathcal{\cap M}^{\prime}\right]\right]|E\right)_{\gamma}
=K_{\mathcal{M}}+\sum_{i}\frac{1}{2^{K_{\mathcal{M}}}}H\left(\mathcal{M}^{\prime}\backslash\left[\mathcal{X}^{\prime}\mathcal{\cap M}^{\prime}\right]|E\right)_{\gamma^{i}}
\end{equation}
Consider that
\begin{equation}
\widetilde{E}_{\operatorname{sq}}(\mathcal{G})_{\gamma_{\mathcal{M}}}\geq\frac{1}{2}\left\vert \mathcal{A}\left(\mathcal{M},\mathcal{G}\right)\right\vert K_{\mathcal{M}}.\label{eq:tilde-I-private-state}
\end{equation}
This is because
\begin{align}
\widetilde{I}(\mathcal{GG}^{\prime}|E)_{\gamma_{\mathcal{MM}^{\prime}E}} & =\sum_{\mathcal{X\in G}}H\left(\left[\mathcal{M}\backslash\left[\mathcal{X\cap M}\right]\right]\left[\mathcal{M}^{\prime}\backslash\left[\mathcal{X}^{\prime}\mathcal{\cap M}^{\prime}\right]\right]|E\right)_{\gamma}\nonumber \\
 & \ \ \ \ -\left(\left\vert \mathcal{A}\left(\mathcal{M},\mathcal{G}\right)\right\vert -1\right)H\left(\mathcal{M}\mathcal{M}^{\prime}|E\right)_{\gamma}\\
 & =\left\vert \mathcal{A}\left(\mathcal{M},\mathcal{G}\right)\right\vert K_{\mathcal{M}}+\sum_{\mathcal{X\in G}}\sum_{i}\frac{1}{2^{K_{\mathcal{M}}}}H\left(\mathcal{M}^{\prime}\backslash\left[\mathcal{X}^{\prime}\mathcal{\cap M}^{\prime}\right]|E\right)_{\gamma^{i}}\nonumber \\
 & \ \ \ \ -\left(\left\vert \mathcal{A}\left(\mathcal{M},\mathcal{G}\right)\right\vert -1\right)\sum_{i}\frac{1}{2^{K_{\mathcal{M}}}}H\left(\mathcal{M}^{\prime}|E\right)_{\gamma^{i}}\\
 & =\left\vert \mathcal{A}\left(\mathcal{M},\mathcal{G}\right)\right\vert K_{\mathcal{M}}+\sum_{i}\frac{1}{2^{K_{\mathcal{M}}}}\widetilde{I}(\mathcal{G}^{\prime}|E)_{\gamma^{i}}\\
 & \geq\left\vert \mathcal{A}\left(\mathcal{M},\mathcal{G}\right)\right\vert K_{\mathcal{M}}
\end{align}
The reasons for these steps are similar to those used to justify (\ref{eq:squash-justify-1})-(\ref{eq:squash-justify-last}),
and we can then conclude (\ref{eq:tilde-I-private-state}). So this
implies that
\begin{equation}
\widetilde{E}_{\operatorname{sq}}(\mathcal{G})_{\Psi_{\mathcal{S}}}\geq\frac{1}{2}\sum_{\mathcal{M\in C}(\mathcal{G})}\left\vert \mathcal{A}\left(\mathcal{M},\mathcal{G}\right)\right\vert \left(K_{\mathcal{M}}+E_{\mathcal{M}}\right)\label{eq:uppbound 2}
\end{equation}
Equations \eqref{eq:upper bound1} and \eqref{eq:uppbound 2} together
conclude the proof.
\end{proof}

\section{Entanglement distillation and secret key agreement using a quantum
broadcast channel}

\label{sub:Quantum-broadcast-channels}A quantum broadcast channel
is a CPTP map $\mathcal{N}_{A\rightarrow B_{1}\cdots B_{m}}$ from
one sender $A$\ to multiple receivers $B_{1}$, \ldots{}, $B_{m}$.
Several communication tasks have already been considered for a quantum
broadcast channel \cite{YHD2006,DHL10,SW11,RSW14}, and in classical
information theory, the secret-key agreement capacity of certain classes
of point-to-multipoint noisy discrete memoryless channels has been
characterized \cite{CN08} and further generalizations have been obtained
as well \cite{GA10a,GA10b}.

In this work, we are interested in bounding the achievable entanglement
distillation and secret key agreement rates between any subset of
the parties when using a quantum broadcast channel an arbitrarily
large number of times, such that the sender and receivers are allowed
to engage in an arbitrary number of rounds of LOCC\ between each
channel use. It is customary to consider the paradigm of local operations
and classical communication (LOCC) for entanglement distillation and
local operations and public communication (LOPC) for secret key agreement.
However, as mentioned in Section~\ref{sec: locc pc}, the approximate
LOPC distillation of secret key is equivalent to the approximate LOCC
distillation of private states \cite{HHHO09,HHHO05}. Therefore, the
two tasks can be studied together under the common umbrella of LOCC.

Let us now define a general protocol for secret key agreement and
entanglement distillation by using a quantum broadcast channel and
LOCC. For simplicity, we begin by considering the case in which we
have a sender Alice and two receivers Bob and Charlie. The most general
$\left(n,E_{AB},E_{AC},E_{BC},E_{ABC},K_{AB},K_{AC},K_{BC},K_{ABC},\varepsilon\right)$
protocol to distill entanglement and secret key between all subsets
of the parties, where $\left(E_{AB},E_{AC},E_{BC},E_{ABC},K_{AB},K_{AC},K_{BC},K_{ABC}\right)$ denotes a rate tuple, involves the following steps:
\begin{enumerate}
\item Alice, Bob, and Charlie engage in a round of LOCC\ to prepare a state
$\rho_{A_{1}A_{1}^{\prime}B_{1}^{\prime}C_{1}^{\prime}}^{\left(1\right)}$.
Necessarily, this state is separable with respect to the cut $A_{1}A_{1}^{\prime}:B_{1}^{\prime}:C_{1}^{\prime}$.
Set $i=1$.
\item Alice transmits system $A_{i}$ through the broadcast channel $\mathcal{N}_{A_{i}\rightarrow B_{i}C_{i}}\equiv\mathcal{N}_{A\rightarrow BC}$,
leading to the state
\[
\sigma_{A_{i}^{\prime}B_{i}B_{i}^{\prime}C_{i}C_{i}^{\prime}}^{\left(i\right)}\equiv\mathcal{N}_{A_{i}\rightarrow B_{i}C_{i}}\left(\rho_{A_{i}A_{i}^{\prime}B_{i}^{\prime}C_{i}^{\prime}}^{\left(i\right)}\right).
\]
Thus, the primed registers $A_{i}^{\prime}B_{i}^{\prime}C_{i}^{\prime}$
represent ``scratch\textquotedblright \ registers of arbitrary size
that the three parties can use for their local processing.
\item Alice, Bob, and Charlie engage in a round of LOCC, leading to the
state $\rho_{A_{i+1}A_{i+1}^{\prime}B_{i+1}^{\prime}C_{i+1}^{\prime}}^{\left(i+1\right)}$.
Set $i:=i+1$.
\item If $i<n$, then go to step 2. Otherwise,\ Alice, Bob, and Charlie
engage in a final round of LOCC\ to produce a state $\Omega_{ABC}$.
The protocol is depicted in Figure~\ref{fig:protocol}. 
\end{enumerate}
\begin{figure}[ptb]
\centering{}\includegraphics[bb=0bp 0bp 10.107100in 2.405900in,width=6.1713in,height=1.4901in]{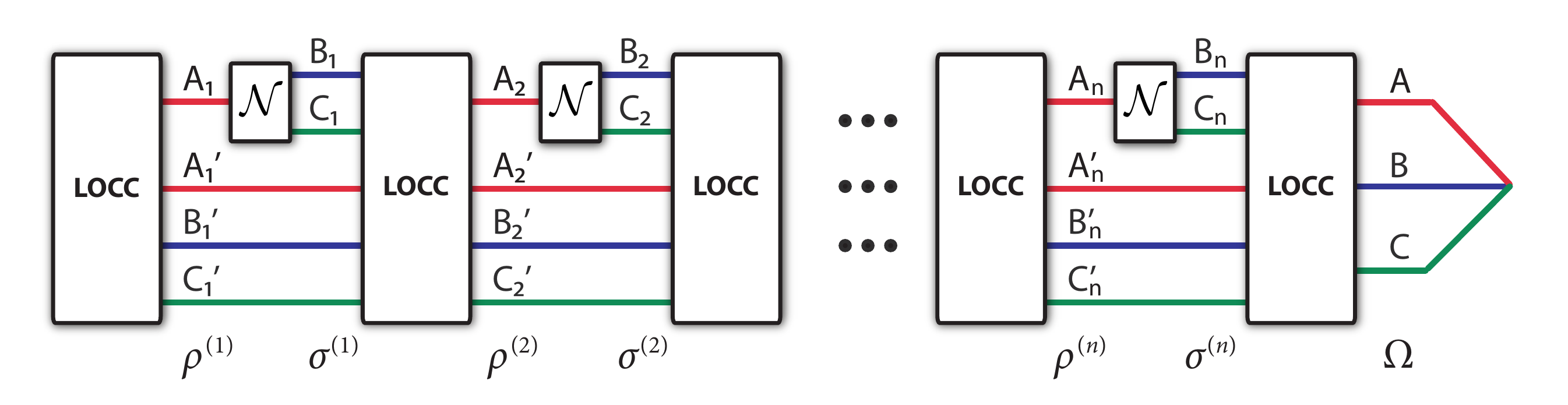}\caption{A general protocol for entanglement distillation and secret key agreement
using LOCC and a quantum broadcast channel $\mathcal{N}_{A\rightarrow BC}$\ with
one sender and two receivers. The protocol uses the channel $n$ times,
and the primed registers represent ``scratch\textquotedblright \ registers
that each party uses for local processing. The state $\Omega_{ABC}$
at the end is $\varepsilon$-close in trace distance to the ideal
state given in~(\ref{eq:example-ideal-state}).}
\label{fig:protocol}
\end{figure}

At the end of the protocol, the state $\Omega_{ABC}$
is $\varepsilon$-close in trace distance to the ideal state $\Psi_{ABC}$\ given
in (\ref{eq:example-ideal-state}):
\begin{equation}
\left\Vert \Omega_{ABC}-\Psi_{ABC}\right\Vert _{1}\leq\varepsilon.\label{eq:protocol condition}
\end{equation}
Furthermore, the entanglement distillation and secret key agreement
rates (similar to (\ref{rate11}) and (\ref{rate12}), but with a factor of $1/n$ to take into account $n$ uses of the channel) are given as
\begin{align}
E_{AB} & \equiv\frac{1}{n}H\left(A_{1}\right)_{\Phi},\ \ \ \ \ E_{AC}\equiv\frac{1}{n}H\left(A_{2}\right)_{\Phi},\ \ \ \ \ E_{BC}\equiv\frac{1}{n}H\left(B_{3}\right)_{\Phi},\ \ \ \ \ E_{ABC}\equiv\frac{1}{n}H\left(A_{4}\right)_{\Phi},\\
K_{AB} & \equiv\frac{1}{n}H\left(A_{5}\right)_{\gamma},\ \ \ \ \ K_{AC}\equiv\frac{1}{n}H\left(A_{6}\right)_{\gamma},\ \ \ \ \ K_{BC}\equiv\frac{1}{n}H\left(B_{7}\right)_{\gamma},\ \ \ \ \ K_{ABC}\equiv\frac{1}{n}H\left(A_{8}\right)_{\gamma},
\end{align}
where the entropies are once again evaluated with respect to the ideal state
in (\ref{eq:example-ideal-state}) and for the private $\gamma$-states,
it is implicit that we are evaluating the entropies of the key systems
(so that the entropy is equal to the number of private bits in the
state).

A rate tuple $\left(E_{AB},E_{AC},E_{BC},E_{ABC},K_{AB},K_{AC},K_{BC},K_{ABC}\right)$\ is
achievable if for all $\varepsilon>0$ and sufficiently large $n$,
there exists an $\left(n,E_{AB},E_{AC},E_{BC},E_{ABC},K_{AB},K_{AC},K_{BC},K_{ABC},\varepsilon\right)$
protocol of the above form. The capacity region is defined to be the
closure of the set of all achievable rates.

The main goal of this paper is to give an outer bound on the capacity
region as defined above. As such, it is helpful to describe the action
of the channel in each round by an isometric extension $U_{A_{i}\rightarrow B_{i}C_{i}E_{i}}^{\mathcal{N}}$,
where $E_{i}$ is an environment system. Including the environment
systems, we then write the state at the conclusion of $i$ steps of
the protocol as $\sigma_{A_{i}^{\prime}B_{i}B_{i}^{\prime}C_{i}C_{i}^{\prime}E^{i}}^{\left(i\right)}$,
where $E^{i}\equiv E_{1}\cdots E_{i}$. It is also helpful to consider
a system $R^{\left(i\right)}$ that purifies the state before the
$i$th channel use, so that
\begin{equation}
\varphi_{A_{i}A_{i}^{\prime}B_{i}^{\prime}C_{i}^{\prime}E^{i-1}R^{\left(i\right)}}^{\rho^{\left(i\right)}}\label{eq:LOCC-input-purify}
\end{equation}
is a purification of $\rho_{A_{i}A_{i}^{\prime}B_{i}^{\prime}C_{i}^{\prime}E^{i-1}}^{\left(i\right)}$.
Let $\sigma_{A_{i}^{\prime}B_{i}B_{i}^{\prime}C_{i}C_{i}^{\prime}E^{i}R^{\left(i\right)}}^{\left(i\right)}$
denote the state which results from applying an isometric extension
$U_{A_{i}\rightarrow B_{i}C_{i}E_{i}}^{\mathcal{N}}$ of the channel
$\mathcal{N}_{A_{i}\rightarrow B_{i}C_{i}}$ to the purification $\varphi_{A_{i}A_{i}^{\prime}B_{i}^{\prime}C_{i}^{\prime}E^{i-1}R^{\left(i\right)}}^{\rho^{\left(i\right)}}$.

The generalization of the above protocol to multiple parties is straightforward,
so we only discuss the main points. The channel is $\mathcal{N}_{A\rightarrow B_{1}\cdots B_{m}}$
and let $\mathcal{S}=\left\{ A,B_{1},\ldots,B_{m}\right\} $. For
a given subset $\mathcal{K\in P}_{\geq2}(\mathcal{S})$,
let $K_{\mathcal{K}}$ denote the rate at which a $\left\vert \mathcal{K}\right\vert $\ multiparty
secret key can be distilled between the members of $\mathcal{K}$,
and let $E_{\mathcal{K}}$ denote the rate at which a $\left\vert \mathcal{K}\right\vert $\ multiparty
GHZ\ entangled state can be distilled between the members of $\mathcal{K}$.
The rate tuple is specified by $\left(E_{\mathcal{K}},K_{\mathcal{K}}\right)_{\mathcal{K\in P}_{\geq2}(\mathcal{S})}$.
After each round of LOCC, the state is $\rho_{A_{i}\mathcal{S}_{i}^{\prime}}^{\left(i\right)}$
and after each channel use, the state is $\sigma_{\left[\mathcal{S}_{i}\backslash A_{i}\right]\mathcal{S}_{i}^{\prime}}^{\left(i\right)}$.
The state generated after the last round of LOCC\ is $\Omega_{\mathcal{S}}$,
which is $\varepsilon$-close to the ideal state $\Psi_{\mathcal{S}}$
given in (\ref{eq:big-entangled-state}). Achievable rates and the
capacity region are defined in a similar way, and it is again helpful
to consider environments resulting from an isometric extension $U_{A\rightarrow B_{1}\cdots B_{m}E}^{\mathcal{N}}$\ of
the channel $\mathcal{N}_{A\rightarrow B_{1}\cdots B_{m}}$.

\section{Bounds on entanglement distillation and secret key agreement for
a two-receiver quantum broadcast channel}

\label{sec:Two Receiver UB}In this section, we establish constraints
on achievable rates for entanglement distillation and secret-key agreement
for a quantum broadcast channel with two receivers. The bounds are
given in terms of the squashed entanglement measures of Section~\ref{sub:Bi and Mul Sq Ent}.

\begin{theorem} \label{thm:two receivers}Let $\mathcal{N}_{A\rightarrow BC}$
be a quantum broadcast channel from a sender Alice to receivers Bob
and Charlie. If the rate tuple $\left(E_{AB},E_{AC},E_{BC},E_{ABC},K_{AB},K_{AC},K_{BC},K_{ABC}\right)$
is achievable, then there exists a pure state $\phi_{RA}$ with
\begin{equation}
\omega_{RBC}\equiv\mathcal{N}_{A\rightarrow BC}(\phi_{RA}),
\end{equation}
such that the following bounds hold
\begin{align}
E_{AB}+K_{AB}+E_{BC}+K_{BC}+E_{ABC}+K_{ABC} & \leq E_{\operatorname{sq}}\left(RC;B\right)_{\omega}\label{eq:AC cut B}\\
E_{AC}+K_{AC}+E_{BC}+K_{BC}+E_{ABC}+K_{ABC} & \leq E_{\operatorname{sq}}\left(RB;C\right)_{\omega}\label{eq:AB cut C}\\
E_{AB}+K_{AB}+E_{AC}+K_{AC}+E_{ABC}+K_{ABC} & \leq E_{\operatorname{sq}}\left(R;BC\right)_{\omega}\label{eq:A cut BC}\\
E_{AB}+K_{AB}+E_{AC}+K_{AC}+E_{BC}+K_{BC}\nonumber \\
+\frac{{3}}{2}\left(E_{ABC}+K_{ABC}\right) & \leq\min\left\{ E_{\operatorname{sq}}\left(R;B;C\right)_{\omega},\widetilde{E}_{\operatorname{sq}}\left(R;B;C\right)_{\omega}\right\} .\label{eq:ABC constr}
\end{align}
The dimension of system $R$ need not be any larger than the dimension
of the channel input. \end{theorem}

\begin{proof} It is important to realize that since we allow all
three parties to participate in each round of LOCC, the bounds we
give on these rates should involve all three of them. Consider an
arbitrary protocol as described in Section~\ref{sub:Quantum-broadcast-channels}.
We work our way backwards through the protocol, starting at the end
and unraveling it until we reach the beginning. The ideal state at
the end of the protocol is $\Psi_{ABC}$, as specified in (\ref{eq:example-ideal-state}),
and the actual state is $\Omega_{ABC}$, as described in Step~4 of
Section~\ref{sub:Quantum-broadcast-channels}. They are related by
(\ref{eq:protocol condition}). Recall the partitions of $\left\{ A,B,C\right\} $\ discussed
in (\ref{eq:partition-ABC-1})-(\ref{eq:partition-ABC-4}).

We begin by considering the constraint in $\left(\ref{eq:AC cut B}\right)$,
which corresponds to the partition $\mathcal{G}_{3}$ in (\ref{eq:partition-ABC-3}).
Consider that
\begin{align}
n\left(E_{AB}+K_{AB}+E_{BC}+K_{BC}+E_{ABC}+K_{ABC}\right) & \leq E_{\operatorname{sq}}\left(AC;B\right)_{\Psi}\\
 & \leq E_{\operatorname{sq}}\left(AC;B\right)_{\Omega}+f_{1}\left(n,\varepsilon\right).
\end{align}
The first inequality follows from (\ref{eq:part-g3}) and the second
follows from an application of the continuity of squashed entanglement
to (\ref{eq:protocol condition}), with $f_{i}\left(n,\varepsilon\right)$
a function such that $\lim_{\varepsilon\searrow0}\lim_{n\rightarrow\infty}\frac{1}{n}f_{i}\left(n,\varepsilon\right)=0$
(we will have more such functions later on). Continuing, we have that
\begin{align}
 & E_{\operatorname{sq}}\left(AC;B\right)_{\Omega}\nonumber \\
 & \leq E_{\operatorname{sq}}\left(A_{n}^{\prime}C_{n}^{\prime}C_{n};B_{n}^{\prime}B_{n}\right)_{\sigma^{\left(n\right)}}\\
 & \leq E_{\operatorname{sq}}\left(A_{n}^{\prime}C_{n}^{\prime}C_{n}B_{n}E_{n};B_{n}^{\prime}\right)_{\sigma^{\left(n\right)}}+E_{\operatorname{sq}}\left(A_{n}^{\prime}C_{n}^{\prime}C_{n}B_{n}^{\prime}E_{1}\cdots E_{n-1}R^{\left(n\right)};B_{n}\right)_{\sigma^{\left(n\right)}}\label{eq:iter-1}\\
 & =E_{\operatorname{sq}}\left(A_{n}^{\prime}C_{n}^{\prime}A_{n};B_{n}^{\prime}\right)_{\rho^{\left(n\right)}}+E_{\operatorname{sq}}\left(A_{n}^{\prime}C_{n}^{\prime}C_{n}B_{n}^{\prime}E_{1}\cdots E_{n-1}R^{\left(n\right)};B_{n}\right)_{\sigma^{\left(n\right)}}\\
 & \leq E_{\operatorname{sq}}\left(A_{n-1}^{\prime}C_{n-1}^{\prime}C_{n-1};B_{n-1}^{\prime}B_{n-1}\right)_{\sigma^{\left(n-1\right)}}+E_{\operatorname{sq}}\left(A_{n}^{\prime}C_{n}^{\prime}C_{n}B_{n}^{\prime}E_{1}\cdots E_{n-1}R^{\left(n\right)};B_{n}\right)_{\sigma^{\left(n\right)}}\label{eq:iter-3}\\
 & \leq\sum_{i=1}^{n}E_{\operatorname{sq}}\left(A_{i}^{\prime}C_{i}^{\prime}C_{i}B_{i}^{\prime}E_{1}\cdots E_{i-1}R^{\left(i\right)};B_{i}\right)_{\sigma^{\left(i\right)}}.\label{eq:two receiver G1 final}
\end{align}
The first inequality follows from monotonicity of the squashed entanglement
under LOCC. The second inequality follows from applying Lemma~\ref{lem:TGW Mult}.
The equality follows from the fact that systems $A_{n}$ and $B_{n}C_{n}E_{n}$
are related by an isometry (i.e., an isometric extension of the channel).
The third inequality is again monotonicity under LOCC. To conclude
the final inequality, we repeat (\ref{eq:iter-1})-(\ref{eq:iter-3})
iteratively. Putting the two inequality chains together, we find that
\begin{align}
E_{AB}+K_{AB} & +E_{BC}+K_{BC}+E_{ABC}+K_{ABC}\nonumber \\
 & \leq\frac{1}{n}\sum_{i=1}^{n}E_{\operatorname{sq}}\left(A_{i}^{\prime}C_{i}^{\prime}C_{i}B_{i}^{\prime}E_{1}\cdots E_{i-1}R^{\left(i\right)};B_{i}\right)_{\sigma^{\left(i\right)}}+\frac{1}{n}f_{1}\left(n,\varepsilon\right)\\
 & =E_{\operatorname{sq}}\left(QSC;B\right)_{\tau}+\frac{1}{n}f_{1}\left(n,\varepsilon\right),\label{eq:app grp lemma}
\end{align}
where
\begin{equation}
\tau_{QSBC}\equiv\sum_{i=1}^{n}\frac{{1}}{n}\vert i\rangle \langle i\vert _{Q}\otimes\mathcal{N}_{A\rightarrow BC}\left(\varphi_{A_{i}^{\prime}B_{i}^{\prime}C_{i}^{\prime}E_{1}\cdots E_{i-1}R^{\left(i\right)}A}^{\left(i\right)}\right),\label{eq:single letter state}
\end{equation}
and $\varphi_{A_{i}^{\prime}B_{i}^{\prime}C_{i}^{\prime}E_{1}\cdots E_{i-1}R^{\left(i\right)}A}^{\left(i\right)}$
is the purification defined in (\ref{eq:LOCC-input-purify}), with
it understood that systems $B_{i}C_{i}$ are isomorphic to systems
$BC$. In the above, $Q$ is a time-sharing or auxiliary classical
system, and $S$ is a register with size
\begin{equation}
\left\vert S\right\vert \geq\max_{i}\left\vert A_{i}^{\prime}C_{i}^{\prime}B_{i}^{\prime}E_{1}\cdots E_{i-1}R^{\left(i\right)}\right\vert ,
\end{equation}
such that it is large enough to contain the largest of the systems
$A_{i}^{\prime}C_{i}^{\prime}B_{i}^{\prime}E_{1}\cdots E_{i-1}R^{\left(i\right)}$
(and simply padded with zeros for smaller systems). Observe that the
state $\tau_{QSBC}$ is constructed from the given protocol. The equality
in $\left(\ref{eq:app grp lemma}\right)$ follows from the application
of Lemma~\ref{lem:sqent}. Thus, we arrive at a single-letter bound.

A similar line of reasoning leads to the following inequalities:
\begin{align}
E_{AC}+K_{AC}+E_{BC}+K_{BC}+E_{ABC}+K_{ABC} & \leq E_{\operatorname{sq}}\left(QSB;C\right)_{\tau}+\frac{1}{n}f_{2}\left(n,\varepsilon\right),\\
E_{AB}+K_{AB}+E_{AC}+K_{AC}+E_{ABC}+K_{ABC} & \leq E_{\operatorname{sq}}\left(QS;BC\right)_{\tau}+\frac{1}{n}f_{3}\left(n,\varepsilon\right),
\end{align}
where we observe that the constraints involve the same state $\tau_{QSBC}$
from (\ref{eq:single letter state}).

We now consider the constraint in $\left(\ref{eq:ABC constr}\right)$.
The reasoning that follows holds for both multipartite squashed entanglements
$E_{\operatorname{sq}}$ and $\widetilde{E}_{\operatorname{sq}}$.
The entanglement distillation and secret key agreement rates of any
protocol can be upper bounded as follows: 
\begin{multline}
n\left(E_{AB}+K_{AB}+E_{AC}+K_{AC}+E_{BC}+K_{BC}+\frac{{3}}{2}\left[E_{ABC}+K_{ABC}\right]\right)\\
\leq E_{\operatorname{sq}}(A;B;C)_{\Psi}\leq E_{\operatorname{sq}}(A;B;C)_{\Omega}+f_{4}\left(n,\varepsilon\right),
\end{multline}
where the first inequality is a consequence of (\ref{eq:G4 E sq})
and the second from applying continuity of squashed entanglement to
(\ref{eq:protocol condition}). Continuing, we find that
\begin{align}
 & E_{\operatorname{sq}}(A;B;C)_{\Omega}\nonumber \\
 & \leq E_{\operatorname{sq}}\left(A_{n}^{\prime};B_{n}^{\prime}B_{n};C_{n}^{\prime}C_{n}\right)_{\sigma^{\left(n\right)}}\\
 & \leq E_{\operatorname{sq}}\left(A_{n}^{\prime}C_{n}B_{n}E_{n};B_{n}^{\prime};C_{n}^{\prime}\right)_{\sigma^{\left(n\right)}}+E_{\operatorname{sq}}\left(A_{n}^{\prime}B_{n}^{\prime}C_{n}^{\prime}E_{1}\cdots E_{n-1}R^{\left(n\right)};B_{n};C_{n}\right)_{\sigma^{\left(n\right)}}\\
 & =E_{\operatorname{sq}}\left(A_{n}^{\prime}A_{n};B_{n}^{\prime};C_{n}^{\prime}\right)_{\rho^{\left(n\right)}}+E_{\operatorname{sq}}\left(A_{n}^{\prime}B_{n}^{\prime}C_{n}^{\prime}E_{1}\cdots E_{n-1}R^{\left(n\right)};B_{n};C_{n}\right)_{\sigma^{\left(n\right)}}\\
 & \leq E_{\operatorname{sq}}\left(A_{n-1}^{\prime};B_{n-1}^{\prime}B_{n-1};C_{n-1}^{\prime}C_{n-1}\right)_{\sigma^{\left(n-1\right)}}+E_{\operatorname{sq}}\left(A_{n}^{\prime}B_{n}^{\prime}C_{n}^{\prime}E_{1}\cdots E_{n-1}R^{\left(n\right)};B_{n};C_{n}\right)_{\sigma^{\left(n\right)}}\\
 & \leq\sum_{i=1}^{n}E_{\operatorname{sq}}\left(A_{i}^{\prime}B_{i}^{\prime}C_{i}^{\prime}E_{1}\cdots E_{i-1}R^{\left(i\right)};B_{i};C_{i}\right)_{\sigma^{\left(i\right)}}.
\end{align}
Putting the above two inequality chains together, we find that
\begin{multline}
E_{AB}+K_{AB}+E_{AC}+K_{AC}+E_{BC}+K_{BC}+\frac{{3}}{2}\left(E_{ABC}+K_{ABC}\right)\\
\leq\frac{1}{n}\sum_{i=1}^{n}E_{\operatorname{sq}}\left(A_{i}^{\prime}B_{i}^{\prime}C_{i}^{\prime}E_{1}\cdots E_{i-1}R^{\left(i\right)};B_{i};C_{i}\right)_{\sigma^{\left(i\right)}}+\frac{1}{n}f_{4}\left(n,\varepsilon\right)\\
=E_{\operatorname{sq}}\left(QS;B;C\right)_{\tau}+\frac{1}{n}f_{4}\left(n,\varepsilon\right),\label{eq:state-rewrite}
\end{multline}
where $\tau$ is defined in $\left(\ref{eq:single letter state}\right)$.
By the same reasoning, we have that
\begin{multline}
E_{AB}+K_{AB}+E_{AC}+K_{AC}+E_{BC}+K_{BC}+\frac{{3}}{2}\left(E_{ABC}+K_{ABC}\right)\\
\leq\widetilde{E}_{\operatorname{sq}}\left(QS;B;C\right)_{\tau}+\frac{1}{n}f_{5}\left(n,\varepsilon\right),
\end{multline}
Note that unlike in the bipartite case, in the multipartite case with
three or more parties, we have the two possible squashed entanglement
measures $\widetilde{E}_{\operatorname{sq}}$ and $E_{\operatorname{sq}}$.
Since in general they are incomparable, either could give a tighter
bound.

The assumption that the rate tuple $\left(E_{AB},E_{AC},E_{BC},E_{ABC},K_{AB},K_{AC},K_{BC},K_{ABC}\right)$\ is
achievable implies that we can take $\varepsilon\searrow0$ as $n\rightarrow\infty$.
So we have shown that the rate tuple satisfies (\ref{eq:AC cut B})-(\ref{eq:ABC constr})
for some input state $\rho_{RA}$ and $\omega_{RBC}\equiv\mathcal{N}_{A\rightarrow BC}(\rho_{RA})$.
Let $\phi_{R^{\prime}RA}^{\rho}$ be a purification of $\rho_{RA}$
and let $\omega_{R^{\prime}RBC}\equiv\mathcal{N}_{A\rightarrow BC}(\phi_{R^{\prime}RA}^{\rho})$.
By monotonicity of squashed entanglement under quantum operations,
we have that
\begin{align}
E_{\operatorname{sq}}\left(RC;B\right)_{\omega} & \leq E_{\operatorname{sq}}\left(R^{\prime}RC;B\right)_{\omega},\\
E_{\operatorname{sq}}\left(RB;C\right)_{\omega} & \leq E_{\operatorname{sq}}\left(R^{\prime}RB;C\right)_{\omega},\\
E_{\operatorname{sq}}\left(R;BC\right)_{\omega} & \leq E_{\operatorname{sq}}\left(R^{\prime}R;BC\right)_{\omega},\\
E_{\operatorname{sq}}\left(R;B;C\right)_{\omega} & \leq E_{\operatorname{sq}}\left(R^{\prime}R;B;C\right)_{\omega},\\
\widetilde{E}_{\operatorname{sq}}\left(R;B;C\right)_{\omega} & \leq\widetilde{E}_{\operatorname{sq}}\left(R^{\prime}R;B;C\right)_{\omega}.
\end{align}
By the Schmidt decomposition, note that we can take $\left\vert RR^{\prime}\right\vert =\left\vert A\right\vert $.
The dimension bound appearing in the statement of theorem comes about
by redefining $R:=RR^{\prime}$. \end{proof}

\section{Bounds on entanglement distillation and secret key agreement for
an $m$-receiver quantum broadcast channel}

\label{sec:m receiver UB}Let $\mathcal{N}_{A\rightarrow B_{1}\cdots B_{m}}$
be a quantum broadcast channel with one sender $A$ and $m$ receivers
$B_{1}$, \ldots{}, $B_{m}$. Let $\mathcal{S}=\left\{ R,B_{1},\ldots,B_{m}\right\} $,
where $R$ is a system that the sender possesses.\ For a given subset
$\mathcal{K\in P}_{\geq2}(\mathcal{S})$, let $K_{\mathcal{K}}$
denote the rate at which a $\left\vert \mathcal{K}\right\vert $\ multiparty
secret key can be distilled between the members of $\mathcal{K}$,
and let $E_{\mathcal{K}}$ denote the rate at which a $\left\vert \mathcal{K}\right\vert $\ multiparty
GHZ\ entangled state can be distilled between the members of $\mathcal{K}$.
We now state our main theorem:

\begin{theorem} \label{thm:m receivers}If the rates $\left(K_{\mathcal{K}},E_{\mathcal{K}}\right)_{\mathcal{K\in P}_{\geq2}(\mathcal{S})}$
are achievable, then there exists a pure state $\phi_{RA}$ with 
\begin{equation}
\omega_{RB_{1}\cdots B_{m}}\equiv\mathcal{N}_{A\rightarrow B_{1}\cdots B_{m}}(\phi_{RA}),
\end{equation}
such that the following bounds hold. For all partitions $\mathcal{G}$
of $\mathcal{S}$, 
\begin{equation}
\frac{1}{2}\sum_{\mathcal{M\in C}(\mathcal{G})}\left\vert \mathcal{A}\left(\mathcal{M},\mathcal{G}\right)\right\vert \left(K_{\mathcal{M}}+E_{\mathcal{M}}\right)\leq\min\left\{ E_{\operatorname{sq}}(\mathcal{G})_{\omega},\ \widetilde{E}_{\operatorname{sq}}(\mathcal{G})_{\omega}\right\} ,
\end{equation}
and
\begin{equation}
\mathcal{A}\left(\mathcal{M},\mathcal{G}\right)\equiv\left\{ X\cap\mathcal{M}\ |\ X\in\mathcal{G}\right\} \backslash\{\emptyset\}.
\end{equation}
The dimension of system $R$ need not be any larger than the dimension
of the channel input. \end{theorem}

\begin{proof} We sketch a proof of this theorem, which proceeds along
the lines of reasoning employed in proving Theorem~\ref{thm:two receivers}.
The proof involves the following steps, again working backwards through
the protocol:
\begin{enumerate}
\item Let $\mathcal{G}$ be a partition of $\mathcal{S}$. The ideal state
at the end of the protocol is $\Psi_{\mathcal{S}}$, as given in (\ref{eq:big-entangled-state}).
Lemma~\ref{lem:max ent state sq ent} establishes the following bound:
\begin{equation}
\frac{1}{2}\sum_{\mathcal{M\in C}(\mathcal{G})}\left\vert \mathcal{A}\left(\mathcal{M},\mathcal{G}\right)\right\vert \left(K_{\mathcal{M}}+E_{\mathcal{M}}\right)\leq\min\left\{ E_{\operatorname{sq}}(\mathcal{G})_{\Psi_{\mathcal{S}}},\ \widetilde{E}_{\operatorname{sq}}(\mathcal{G})_{\Psi_{\mathcal{S}}}\right\} .
\end{equation}

\item The actual state generated by the protocol is $\Omega_{\mathcal{S}}$,
as specified in Section~\ref{sub:Quantum-broadcast-channels}. Use
the fact that $\Omega_{\mathcal{S}}$ is $\varepsilon$-close to $\Psi_{\mathcal{S}}$
and the continuity of squashed entanglement to establish that
\begin{align}
E_{\operatorname{sq}}(\mathcal{G})_{\Psi_{\mathcal{S}}} & \leq E_{\operatorname{sq}}(\mathcal{G})_{\Omega_{\mathcal{S}}}+f_{\mathcal{G}}^{1}\left(n,\varepsilon\right),\\
\widetilde{E}_{\operatorname{sq}}(\mathcal{G})_{\Psi_{\mathcal{S}}} & \leq\widetilde{E}_{\operatorname{sq}}(\mathcal{G})_{\Omega_{\mathcal{S}}}+f_{\mathcal{G}}^{2}\left(n,\varepsilon\right),
\end{align}
for $f_{\mathcal{G}}^{i}\left(n,\varepsilon\right)$ some function
with the property that $\lim_{\varepsilon\searrow0}\lim_{n\rightarrow\infty}\frac{1}{n}f_{\mathcal{G}}^{i}\left(n,\varepsilon\right)=0$.
\item Use the fact that the squashed entanglement is non-increasing under
LOCC.
\item Invoke the subadditivity lemma (Lemma \ref{lem:TGW Mult}).
\item Invert the action of the channel on the $i$th input to replace $B_{i}C_{i}E_{i}\rightarrow A_{i}$
since the systems $B_{i}C_{i}E_{i}$ and $A_{i}$ are related by an
isometric extension of the channel.
\item Iterate Steps 3-5 for every step of the protocol.
\item As in (\ref{eq:state-rewrite}), use Lemma~\ref{lem:sqent} to rewrite
the linear sum of squashed entanglements as the squashed entanglement
of a single state to obain a single letter bound. We can use the same
reasoning as at the end of the proof of Theorem~\ref{thm:two receivers}
to restrict the input state to be pure. 
\end{enumerate}
\noindent That concludes the proof sketch. \end{proof}

\section{Application to a pure-loss bosonic broadcast channel}

\label{sec:BBC}We now apply our results to a pure-loss bosonic broadcast
channel, generalizing prior results for the single-sender single-receiver
case \cite{TGW14IEEE,TGW14Nat}. For simplicity, we consider a one-sender
two-receiver channel from a sender Alice to receivers Bob and Charlie.
However, note that the methods given here can be combined with Theorem~\ref{thm:m receivers}\ to
determine bounds for an arbitrary number of receivers. A pure-loss
bosonic channel of the above type can be modeled as
\begin{align}
\label{eq:QBC_in-out1}
\hat{b} & =  \sqrt{\eta_B} \, \hat{a}' 
+ \sqrt{\frac{\eta_B (1-\eta_B-\eta_C)}{\eta_B+\eta_C}} \, \hat{f} 
+ \sqrt{\frac{\eta_C}{\eta_B+\eta_C}} \, \hat{g} ,
\\
\label{eq:QBC_in-out2}
\hat{c} & =  -\sqrt{\eta_C} \, \hat{a}' 
- \sqrt{\frac{\eta_C (1-\eta_B-\eta_C)}{\eta_B+\eta_C}} \, \hat{f} 
+ \sqrt{\frac{\eta_C}{\eta_B+\eta_C}} \, \hat{g} ,
\\
\label{eq:QBC_in-out3}
\hat{e} & =  -\sqrt{1-\eta_B-\eta_C} \, \hat{a}' 
+ \sqrt{\eta_B+\eta_C} \, \hat{f} ,
\end{align}
where $\hat{a}'$, $\hat{b}$, $\hat{c}$, $\hat{e}$ are 
annihilation operators for Alice's input, Bob's output, and Charlie's output 
modes, respectively, $\hat{f}$ and $\hat{g}$ are annihilation operators 
for vacuum inputs from the environment, and the $\eta_B,\ \eta_C>0$ are transmission coefficients such that $\eta_B+\eta_C\leq 1$. The model generalizes the bosonic
broadcast channel from prior work \cite{GS07,GSE07}, in that all
of the light does not necessarily make it to the two receivers and
that which does not is given to the eavesdropper.

In any protocol for entanglement distillation and secret key agreement,
we assume that the final step of the protocol outputs a finite-dimensional
state, i.e., the goal is to generate maximally entangled states of
finite Schmidt rank and finite-dimensional private states. This is
a common approach in continuous-variable quantum information theory
\cite{HW01,GGLMSY04,WHG11,WHG12}, simply because both quantum capacity
and private capacity are measured in qubits and private bits per channel
use, respectively. This approach furthermore provides a mathematical
convenience:\ the only aspect of our analysis here which requires
finite-dimensional states is when we apply continuity of squashed
entanglement at the end of the protocol. All other steps rely on properties
of entropy or the quantum data processing inequality, which is known
to hold in very general settings \cite{U77}. Furthermore, we begin
by assuming that each channel input has a mean photon number constraint
$\left\langle \hat{a}^{\dagger}\hat{a}\right\rangle \leq N_{S}$ for
some $N_{S}$ such that $0\leq N_{S}<\infty$, but we eventually take
a limit as $N_{S}\rightarrow\infty$, indicating that our bounds are
photon-number independent as is the case in \cite{TGW14IEEE,TGW14Nat}.

We then need to determine the (multipartite) squashed entanglement.
In this regard, it is not necessarily an easy task to optimize over
all possible squashing channels of Eve. However, since any squashing
channel can be used to give an upper bound for the rates, we choose
to optimize over squashing channels that are also pure-loss bosonic
channels, modeled by a beamsplitter given by the following mode transformation:
\begin{equation}
\hat{e}\rightarrow\sqrt{\eta_{E^{\prime}}}\hat{e}^{\prime}+\sqrt{1-\eta_{E^{\prime}}}\hat{f}^{\prime},\quad\eta_{E^{\prime}}\in\left[0,1\right].\label{eq:bosonic-squashing-channel}
\end{equation}

We begin by stating and proving the following proposition, which is
more general than what we need, but the proof indicates a general
approach that we employ to establish the main theorem of this section.

\begin{proposition} \label{prop:bosonic-multi-squash} Let $\mathcal{N}_{A\rightarrow B_1B_2\ldots B_m}$ be a pure-loss bosonic
broadcast channel from a sender $A$ to $m$ receivers $B_{1}$,
\ldots{}, $B_{m}$ with transmission coefficients $\eta_i\geq 0$ for all $i\in\left\{ 1,\ldots,m\right\} $, such that 
\begin{equation}
\eta\equiv\sum_{i=1}^{m}\eta_{i}\leq 1.
\end{equation}
Then the following upper bound holds for the squashed entanglements
of the bosonic broadcast channel
\begin{multline}
\max\left\{ \sup_{\phi_{RA}}E_{\operatorname{sq}}\left(R;B_{1};\cdots;B_{m}\right)_{\omega},\ \sup_{\phi_{RA}}\widetilde{E}_{\operatorname{sq}}\left(R;B_{1};\cdots;B_{m}\right)_{\omega}\right\} \\
\leq\frac{1}{2}\left[\sum_{i=1}^{m}\log\left(\frac{\eta_{i}}{\left(1-\eta\right)\eta_{E^{\prime}}^{\ast}}+1\right)+\log\left(\frac{\eta}{\left(1-\eta\right)\left(1-\eta_{E^{\prime}}^{\ast}\right)}+1\right)\right],\label{eq:general-bosonic-bound}
\end{multline}
where
\begin{equation}
\omega_{RB_{1}\cdots B_{m}}\equiv\mathcal{N}_{A\rightarrow B_{1}\cdots B_{m}}\left(\phi_{RA}\right),
\end{equation}
and $\eta_{E^{\prime}}^{\ast}$
is the solution of
\begin{equation}
\sum_{i=1}^{m}\frac{1}{\eta_{E^{\prime}}^{2}\left(1-\eta\right)/\eta_{i}+\eta_{E^{\prime}}}=\frac{1}{\left(1-\eta_{E^{\prime}}\right)^{2}\left(1-\eta\right)/\eta+1-\eta_{E^{\prime}}}.
\end{equation}

\end{proposition}

\begin{proof} Our proof of this proposition generalizes the proof
of \cite[Eq.~(27)]{TGW14IEEE}. Let $\varphi_{RB_{1}\cdots B_{m}E}$
be the pure state that results from applying the channel to a pure
state input $\phi_{RA}$ satisfying $\left\langle \hat{a}^{\dagger}\hat{a}\right\rangle _{\phi_{A}}\leq N_{S}$.
Let $\varphi_{RB_{1}\cdots B_{m}E^{\prime}F^{\prime}}$ be the state resulting
from applying the squashing transformation in (\ref{eq:bosonic-squashing-channel})
to the system $E$ of $\varphi_{RB_{1}\cdots B_{m}E}$. Then
\begin{align}
2E_{\operatorname{sq}}\left(R;B_{1};\cdots;B_{m}\right)_{\omega} & \leq H\left(R|E^{\prime}\right)_{\varphi}+\sum_{i=1}^{m}H\left(B_{i}|E^{\prime}\right)_{\varphi}-H\left(RB_{1}\cdots B_{m}|E^{\prime}\right)_{\varphi}\\
 & =\sum_{i=1}^{m}H\left(B_{i}|E^{\prime}\right)_{\varphi}-H\left(B_{1}\cdots B_{m}|RE^{\prime}\right)_{\varphi}\\
 & =\sum_{i=1}^{m}H\left(B_{i}|E^{\prime}\right)_{\varphi}+H\left(B_{1}\cdots B_{m}|F^{\prime}\right)_{\varphi}
\end{align}
As written, the conditional entropies in the last line are now functions
of the input density operator $\phi_{A}$. Applying the extremality
of Gaussian states for the conditional entropy \cite{EW07,WGC06},
we can conclude that these quantities are all optimized by a thermal
state of mean photon number $N_{S}$. For such a state, one can work
out using the symplectic formalism for bosonic states \cite{WPGCRSL12}
that
\begin{align}
H\left(B_{i}E^{\prime}\right) & =g\left(\left[\eta_{i}+\left(1-\eta\right)\eta_{E^{\prime}}\right]N_{S}\right),\\
H\left(E^{\prime}\right) & =g\left(\left(1-\eta\right)\eta_{E^{\prime}}N_{S}\right),\\
H\left(B_{1}\cdots B_{m}F^{\prime}\right) & =g\left(\left[\eta+\left(1-\eta\right)\left(1-\eta_{E^{\prime}}\right)\right]N_{S}\right),\\
H\left(F^{\prime}\right) & =g\left(\left(1-\eta\right)\left(1-\eta_{E^{\prime}}\right)N_{S}\right),
\end{align}
where
\begin{equation}
g(x)\equiv\left(x+1\right)\log_{2}\left(x+1\right)-x\log_{2}x
\end{equation}
is the entropy of a thermal state of mean photon number $x$. Each
entropy above can be understood in a simple way:\ for a simple pure loss bosonic broadcast
channel $\mathcal{N}_{A\rightarrow B_1B_2\ldots B_m}$, each
state held by any group of parties at the receiving end is unitarily
equivalent to a thermal state with mean photon number proportional
to the fraction of light that makes it to them. So this leads to the
photon-number dependent upper bound
\begin{multline}
2E_{\operatorname{sq}}\left(R;B_{1};\cdots;B_{m}\right)_{\omega}\leq\sum_{i=1}^{m}\bigg( g\left(\left[\eta_{i}+\left(1-\eta\right)\eta_{E^{\prime}}\right]N_{S}\right)-g\left(\left(1-\eta\right)\eta_{E^{\prime}}N_{S}\right)\bigg)\\
+g\left(\left[\eta+\left(1-\eta\right)\left(1-\eta_{E^{\prime}}\right)\right]N_{S}\right)-g\left(\left(1-\eta\right)\left(1-\eta_{E^{\prime}}\right)N_{S}\right).
\end{multline}
One can easily compute the derivative of $g(x)-g\left(\lambda x\right)$
to show that this function is monotonically increasing in $x$ for
$x\geq0$ and $\lambda\in\left[0,1\right]$, and furthermore, one
can easily show that
\begin{equation}
\lim_{x\rightarrow\infty}g(x)-g\left(\lambda x\right)=\log\left(1/\lambda\right).
\end{equation}
So we can conclude that the right hand side (RHS)\ above is a monotonically
increasing function of $N_{S}\geq0$ and taking the limit $N_{S}\rightarrow\infty$
only increases the upper bound. This leads to the following photon-number
independent upper bound:
\begin{equation}
2E_{\operatorname{sq}}\left(R;B_{1};\cdots;B_{m}\right)_{\omega}\leq\sum_{i=1}^{m}\log\left(\frac{\eta_{i}}{\left(1-\eta\right)\eta_{E^{\prime}}}+1\right)+\log\left(\frac{\eta}{\left(1-\eta\right)\left(1-\eta_{E^{\prime}}\right)}+1\right).\label{eq:sq-bound-bosonic}
\end{equation}
which holds for arbitrary $\eta_{E^{\prime}}\in\left[0,1\right]$.
To get the tightest upper bound, we should minimize the RHS of (\ref{eq:sq-bound-bosonic})\ with
respect to $\eta_{E^{\prime}}$. Any local minimum of this function
is a global minimum because the function $\log\left(1+a/x\right)$
is convex in $x$ for $a\geq0$ and $x\geq0$ (as can be checked by
computing the second derivative) and the RHS\ of (\ref{eq:sq-bound-bosonic})
is convex in $\eta_{E^{\prime}}$ as it is a sum of convex functions.
Since we need to solve for $\eta_{E^{\prime}}$ in 
\begin{equation}
\frac{\partial}{\partial\eta_{E^{\prime}}}\left[\sum_{i=1}^{m}\log\left(\frac{\eta_{i}}{\left(1-\eta\right)\eta_{E^{\prime}}}+1\right)+\log\left(\frac{\eta}{\left(1-\eta\right)\left(1-\eta_{E^{\prime}}\right)}+1\right)\right]=0,
\end{equation}
we can use that the first derivative of $\log\left(1+a/x\right)$
is equal to $-1/\left(x^{2}/a+x\right)$, which leads to solving the
following equation for $\eta_{E^{\prime}}$:
\begin{equation}
\sum_{i=1}^{m}\frac{1}{\eta_{E^{\prime}}^{2}\left(1-\eta\right)/\eta_{i}+\eta_{E^{\prime}}}=\frac{1}{\left(1-\eta_{E^{\prime}}\right)^{2}\left(1-\eta\right)/\eta+1-\eta_{E^{\prime}}}.
\end{equation}
This establishes one of the inequalities in (\ref{eq:general-bosonic-bound}).

By a similar line of reasoning as above, consider that
\begin{align}
 & 2\widetilde{E}_{\operatorname{sq}}\left(R;B_{1};\cdots;B_{m}\right)_{\omega}\nonumber \\
 & \leq H\left(RB_{1}\cdots B_{m}|E^{\prime}\right)_{\varphi}-H\left(R|B_{1}\cdots B_{m}E^{\prime}\right)_{\varphi}-\sum_{i=1}^{m}H\left(B_{i}|RB_{\left[m\right]\backslash\left\{ i\right\} }E^{\prime}\right)_{\varphi}\\
 & =H\left(B_{1}\cdots B_{m}|E^{\prime}\right)_{\varphi}+\sum_{i=1}^{m}H\left(B_{i}|F^{\prime}\right)_{\varphi}\label{eq:tilde-E-to-bound-bosonic}
\end{align}
Here again we have written the entropies as a function of the input
density operator, which we know from the extremality of Gaussian states
is optimized by a thermal state for a fixed photon number. For such
an input, we have that
\begin{align}
H\left(B_{1}\cdots B_{m}E^{\prime}\right) & =g\left(\left[\eta+\left(1-\eta\right)\eta_{E^{\prime}}\right]N_{S}\right),\\
H\left(E^{\prime}\right) & =g\left(\left(1-\eta\right)\eta_{E^{\prime}}N_{S}\right),\\
H\left(B_{i}F{\prime}\right) & =g\left(\left[\eta_{i}+\left(1-\eta\right)\left(1-\eta_{E^{\prime}}\right)\right]N_{S}\right),\\
H\left(F{\prime}\right) & =g\left(\left(1-\eta\right)\left(1-\eta_{E^{\prime}}\right)N_{S}\right),
\end{align}
so that (\ref{eq:tilde-E-to-bound-bosonic}) is bounded from above
by
\begin{multline}
g\left(\left[\eta+\left(1-\eta\right)\eta_{E^{\prime}}\right]N_{S}\right)-g\left(\left(1-\eta\right)\eta_{E^{\prime}}N_{S}\right)\\
+\sum_{i=1}^{m}\bigg(g\left(\left[\eta_{i}+\left(1-\eta\right)\left(1-\eta_{E^{\prime}}\right)\right]N_{S}\right)-g\left(\left(1-\eta\right)\left(1-\eta_{E^{\prime}}\right)N_{S}\right)\bigg)\\
\leq\sum_{i=1}^{m}\log\left(\frac{\eta_{i}}{\left(1-\eta\right)\left(1-\eta_{E^{\prime}}\right)}+1\right)+\log\left(\frac{\eta}{\left(1-\eta\right)\eta_{E^{\prime}}}+1\right).\label{eq:tilde-sq-bound-bosonic}
\end{multline}
This bound applies for every $\eta_{E^{\prime}}\in\left[0,1\right]$,
so we can take a minimum over all such $\eta_{E^{\prime}}$. However,
we can now observe that this minimum is exactly the same as the one
above because the RHS\ of (\ref{eq:tilde-sq-bound-bosonic}) is related
to the RHS\ of (\ref{eq:sq-bound-bosonic}) by $\eta_{E^{\prime}}\leftrightarrow1-\eta_{E^{\prime}}$.
\end{proof}

\bigskip{}

We now state the main theorem of this section, which bounds the entanglement
distillation and secret key agreement rates achievable with a pure-loss
bosonic broadcast channel that has one sender and two receivers.

\begin{theorem} \label{thm:two receivers-PLBosonic}Let a pure-loss
bosonic broadcast channel from a sender Alice to receivers Bob and
Charlie be described by the mode transformations in \eqref{eq:QBC_in-out1}-\eqref{eq:QBC_in-out3}.
Then the achievable entanglement distillation and secret key agreement
rates (see Section~\ref{sub:Quantum-broadcast-channels}) are bounded
as follows:
\begin{align}
E_{AB}+K_{AB}+E_{BC}+K_{BC}+E_{ABC}+K_{ABC} & \leq\log\left(\frac{1+\eta_{B}-\eta_{C}}{1-\eta_{B}-\eta_{C}}\right),\\
E_{AC}+K_{AC}+E_{BC}+K_{BC}+E_{ABC}+K_{ABC} & \leq\log\left(\frac{1+\eta_{C}-\eta_{B}}{1-\eta_{B}-\eta_{C}}\right),\\
E_{AB}+K_{AB}+E_{AC}+K_{AC}+E_{ABC}+K_{ABC} & \leq\log\left(\frac{1+\eta_{B}+\eta_{C}}{1-\eta_{B}-\eta_{C}}\right),
\end{align}
and
\begin{multline}
E_{AB}+K_{AB}+E_{AC}+K_{AC}+E_{BC}+K_{BC}+\frac{{3}}{2}\left(E_{ABC}+K_{ABC}\right)\\
\leq\frac{1}{2}\left[\log\left(\frac{\eta_{B}}{\left(1-\eta\right)\left(1-\eta_{E^{\prime}}^{\ast}\right)}+1\right)+\log\left(\frac{\eta_{C}}{\left(1-\eta\right)\left(1-\eta_{E^{\prime}}^{\ast}\right)}+1\right)+\log\left(\frac{\eta}{\left(1-\eta\right)\eta_{E^{\prime}}^{\ast}}+1\right)\right],
\end{multline}
where $\eta_{E^{\prime}}^{\ast}$ is the solution of
\begin{equation}
\frac{1}{\eta_{E^{\prime}}^{2}\left(1-\eta\right)/\eta_{B}+\eta_{E^{\prime}}}+\frac{1}{\eta_{E^{\prime}}^{2}\left(1-\eta\right)/\eta_{C}+\eta_{E^{\prime}}}=\frac{1}{\left(1-\eta_{E^{\prime}}\right)^{2}\left(1-\eta\right)/\eta+1-\eta_{E^{\prime}}}.
\end{equation}

\end{theorem}

\begin{proof} Here we only highlight the main steps without giving
reasons, as much of it is the same as in the proof of Proposition~\ref{prop:bosonic-multi-squash}.
Our approach is simply to bound the quantities $E_{\operatorname{sq}}\left(RC;B\right)_{\omega}$,
$E_{\operatorname{sq}}\left(RB;C\right)_{\omega}$, $E_{\operatorname{sq}}\left(R;BC\right)_{\omega}$,
$E_{\operatorname{sq}}\left(R;B;C\right)_{\omega}$ from Theorem~\ref{thm:two receivers}.
Consider that
\begin{align}
2E_{\operatorname{sq}}\left(RC;B\right)_{\omega} & \leq I\left(RC;B|E^{\prime}\right)\\
 & =H\left(B|E^{\prime}\right)-H\left(B|RCE^{\prime}\right)\\
 & =H\left(B|E^{\prime}\right)+H\left(B|F^{\prime}\right)\\
 & \leq\log\left(\frac{\eta_{B}}{\left(1-\eta\right)\eta_{E^{\prime}}}+1\right)+\log\left(\frac{\eta_{B}}{\left(1-\eta\right)\left(1-\eta_{E^{\prime}}\right)}+1\right)\label{eq:last-line-bosonic}
\end{align}
Picking $\eta_{E^{\prime}}=1/2$ then gives the bound:
\begin{equation}
E_{\operatorname{sq}}\left(RC;B\right)_{\omega}\leq\log\left(\frac{2\eta_{B}}{1-\eta}+1\right)=\log\left(\frac{2\eta_{B}+1-\eta}{1-\eta}\right)=\log\left(\frac{1+\eta_{B}-\eta_{C}}{1-\eta_{B}-\eta_{C}}\right)
\end{equation}
This is optimal because the function in (\ref{eq:last-line-bosonic})
is convex in $\eta_{E^{\prime}}$ and symmetric about $\eta_{E^{\prime}}=1/2$.

Similarly, we have
\begin{equation}
E_{\operatorname{sq}}\left(RB;C\right)_{\omega}\leq\log\left(\frac{1+\eta_{C}-\eta_{B}}{1-\eta_{B}-\eta_{C}}\right),\ \ \ \ \ \ \ \ \ \ E_{\operatorname{sq}}\left(R;BC\right)_{\omega}\leq\log\left(\frac{1+\eta_{B}+\eta_{C}}{1-\eta_{B}-\eta_{C}}\right).
\end{equation}
Furthermore, we can apply Proposition~\ref{prop:bosonic-multi-squash}
to find that
\begin{multline}
2E_{\operatorname{sq}}\left(R;B;C\right)_{\omega}\leq\log\left(\frac{\eta_{B}}{\left(1-\eta\right)\left(1-\eta_{E^{\prime}}^{\ast}\right)}+1\right)\\
+\log\left(\frac{\eta_{C}}{\left(1-\eta\right)\left(1-\eta_{E^{\prime}}^{\ast}\right)}+1\right)+\log\left(\frac{\eta}{\left(1-\eta\right)\eta_{E^{\prime}}^{\ast}}+1\right),
\end{multline}
where $\eta_{E^{\prime}}^{\ast}$ is the solution of
\begin{equation}
\frac{1}{\eta_{E^{\prime}}^{2}\left(1-\eta\right)/\eta_{B}+\eta_{E^{\prime}}}+\frac{1}{\eta_{E^{\prime}}^{2}\left(1-\eta\right)/\eta_{C}+\eta_{E^{\prime}}}=\frac{1}{\left(1-\eta_{E^{\prime}}\right)^{2}\left(1-\eta\right)/\eta+1-\eta_{E^{\prime}}}.
\end{equation}
This completes the proof.
\end{proof}

\section{Conclusion}

\label{sec:DC}We have shown how multipartite generalizations of the
squashed entanglement \cite{YHHHOS09,AHS08} lead to several constraints
on the rates at which secret key and entanglement can be generated
between any subset of the users of a quantum broadcast channel. Along
the way, we developed several new properties of these measures, which
include the subadditivity lemma (Lemma~\ref{lem:TGW Mult}), monotonicity
under groupings, reductions for product states, and the evaluation
of the measures for a tensor product of entangled and private states
shared between all subsets of a given set of parties. Finally, we
applied our results to a single-sender two-receiver bosonic broadcast
channel.

Some future directions include to determine upper bounds on
the secret key agreement and entanglement distillation capacity of
a multiple access or more general quantum network channel. One could
also attempt the challenging task of proving that the bounds given
here are strong converse rates. However, it is not yet known whether
the single-sender single-receiver squashed entanglement is a strong
converse rate.

\bigskip{}

\textbf{Acknowledgements.} We are grateful to Saikat Guha for discussions
related to the topic of this paper. KS acknowledges support from NSF\ Grant
No.~CCF-1350397, the DARPA Quiness Program through US Army Research
Office award W31P4Q-12-1-0019, and the Graduate School of Louisiana
State University for the 2014-2015 Dissertation Year Fellowship. MMW
acknowledges support from startup funds from the Department of Physics
and Astronomy at LSU, the NSF\ under Award No.~CCF-1350397, and
the DARPA Quiness Program through US Army Research Office award W31P4Q-12-1-0019.
MT acknowledges support from Open Partnership Joint Projects of JSPS
Bilateral Joint Research Projects and ImPACT Program of Council for
Science, Technology and Innovation, Japan. He is also grateful to
members of the Hearne Institute for Theoretical Physics at LSU for
their hospitality during his research visit in February 2015.

 \bibliographystyle{plain}
\bibliography{MSKED-Ref}

\end{document}